\documentclass[11pt]{article}

\usepackage[T1]{fontenc}
\usepackage[utf8]{inputenc}
\usepackage{lmodern}
\usepackage{xcolor}
\usepackage{amssymb}
\usepackage{amsmath, amssymb, amsfonts, amsthm, mathrsfs}

\theoremstyle{plain}
\newtheorem{lemma}{Lemma}[section]           
\newtheorem{proposition}{Proposition}[section] 

\newtheoremstyle{uprightremark} 
  {}                             
  {}                             
  {\normalfont}                  
  {}                             
  {\bfseries}                    
  {.}                            
  { }                            
  {}                             
\theoremstyle{uprightremark}    
\newtheorem{remark}{Remark}[section]

\theoremstyle{definition}
\newtheorem{assumption}{Assumption}[section]

\usepackage[ruled,vlined]{algorithm2e}
\SetKwRepeat{Do}{do}{while}%
\SetKwInput{KwInit}{Initialization}%
\SetKwIF{If}{ElseIf}{Else}{if}{then}{else if}{else}{endif}%

\usepackage{graphicx}
\usepackage{subcaption}   
\usepackage[section]{placeins}

\usepackage{natbib}
\usepackage{titlesec}
\usepackage{enumitem}
\usepackage{authblk}
\usepackage{multirow}
\usepackage{array}
\usepackage{hyperref}
\usepackage{soul}         
\usepackage{comment}

\newcommand{\vs}{\vspace*{3mm}}

\newcommand{\argmin}{\mathop{\rm argmin}}

\def\blot{\quad {$\vcenter{\vbox{\hrule height.4pt
             \hbox{\vrule width.4pt height.9ex \kern.9ex \vrule
width.4pt}
             \hrule height.4pt}}$}}

\usepackage{multirow}
\providecommand{\keywords}[1]
{
{  \small	
  {\textbf{Keywords---}} #1}
}
\usepackage{cleveref}

\begin{document}

\title{Joint modeling and inference of multiple-subject \\
high-dimensional sparse vector autoregressive models}
\author[1]{Younghoon Kim$^*$}
\author[2]{Zachary F. Fisher}
\author[3]{Vladas Pipiras}
\affil[1]{Department of Mathematics, University of Alabama}
\affil[2]{Department of Psychology and Neuroscience, University of North Carolina at Chapel Hill}
\affil[3]{Department of Statistics and Operations Research, University of North Carolina at Chapel Hill}
\def\thefootnote{$*$}\footnotetext{Corresponding author. Email: ykim124@ua.edu}

\date{\today}

\maketitle
\begin{abstract}
The multiple-subject vector autoregression (multi-VAR) model captures heterogeneous temporal dependence across subjects through commonly shared and subject-unique paths (entries) of individual (subject-specific) sparse VAR transition matrices. 
Available joint penalized regression approaches suffer from statistical inefficiency, whereby the convergence rates of the individual estimators are determined by a subject with the largest sparsity level.
An enhanced multi-VAR estimation approach is proposed based on a communication-efficient data integration framework and a suitable identifiability condition in meta-analysis. 
It is shown that component-wise different convergence rates can be achieved.
A hypothesis testing framework is also developed to assess the significance of common paths, as well as the homogeneity of subject-specific paths. 
Furthermore, a multiple hypothesis test to control the false discovery rate is suggested within a subregion of interest among the entire paths.
A simulation study and a real data application compare the performance of the considered methods to other approaches.
\end{abstract}


\baselineskip18pt


\section{Introduction}

\subsection{Multiple-Subject High-Dimensional Time Series Models}

Sparse vector autoregressive (VAR) modeling of high-dimensional time series (HDTS) has become an important research topic, driven by the increasing availability of high-dimensional, temporally dependent data. These data are collected across diverse scientific domains, including neuroscience, genomics, finance, and social networks \cite[e.g.,][]{kock2015oracle,davis2016sparse}. This trend is particularly prominent in the behavioral and biomedical sciences with the rise of intensive longitudinal data \cite[ILD; e.g.,][]{bolger2013intensive}. A key challenge in analyzing these data is characterizing the dynamic dependencies among a large number of variables, often through the framework of network Granger causality in high-dimensional regimes \citep{shojaie2022granger}, which encodes directional relationships, with edges indicating whether past values of one variable improve the prediction of another.

Extending VAR models to multiple subjects has also received some attention. Joint VAR modeling can reveal shared mechanisms across subjects and enable comparisons of heterogeneous dynamics. Related ideas have been explored in multi-level modeling, which is widely used in individual-differences analyses \cite[e.g.,][]{jongerling2015multilevel,haslbeck2025testing}. A promising high-dimensional extension is the multiple-subject VAR model \cite[\emph{multi-VAR};][]{fisher2022penalized,fisher2024structured}. Instead of relying on mixed-effects formulations, multi-VAR assumes that the sparsely estimated entries of individual VAR transition matrices, called \emph{paths}, can be decomposed into common (or shared) and unique (or subject-unique) components. Specifically, suppose we fit data from $K(>1)$ number of subjects using VAR$(p)$, $p\geq1$, model. For the $K$-stacked $d$-dimensional transition matrices $B^{(k)} = ((\Phi_1^{(k)})' \ldots (\Phi_p^{(k)})')' \in\mathbb{R}^{pd \times d}$, $k\in\{1,\ldots,K\}=:[K]$, consider its vectorization $\beta^{(k)}=\textrm{vec}(B^{(k)})\in\mathbb{R}^{d^2p}$. The multi-VAR model assumes that, for $k\in[K]$, the so-called \emph{individual paths} are decomposed into
\begin{equation}\label{e:decomposition}
    \beta^{(k)} = \alpha^{(0)} + \alpha^{(k)},
\end{equation}
where $\alpha^{(0)}$ are the \emph{common paths} shared by all $K$ subjects and $\alpha^{(k)}$ are \emph{unique paths} of $k^{\textrm{th}}$ subject. The model encourages similarity across individual parameter vectors while allowing for deviations, thereby capturing both shared temporal dynamics and subject-unique variation. Compared with other VAR-type models commonly used in psychology and neuroscience \cite[e.g.,][]{chen2011vector,gates2012group,seth2015granger}, multi-VAR explicitly borrows information across subjects rather than fitting each subject separately, while still accounting for individual differences.

Note that model \eqref{e:decomposition} is inherently not identifiable. Consequently, identification restrictions must be imposed to uniquely determine the parameters; ideally, if the majority of the entries for a given path of $\beta^{(k)}$ across $k$ are zero, this sparsity pattern should be reflected in the corresponding path of $\alpha^{(0)}$ and $\alpha^{(k)}$. Furthermore, it is desirable to maintain statistical efficiency, ensuring that the estimator for $\alpha^{(0)}$ converges to its true value at a faster rate than that of the subject-specific parameters $\beta^{(k)}$. However, the estimation methods applied in the multi-VAR framework \citep{fisher2021multivar,fisher2024structured} fail to preserve these properties. While these approaches provide a well-defined solution, the resulting estimators for both the common and unique components remain statistically inefficient \citep{asiaee2019data} (see also \cite{maity2022meta} for a related discussion).

In addition to these estimation challenges, the literature still lacks a formal hypothesis testing framework for multi-subject VAR models. Rigorous statistical tests are needed to assess whether joint modeling is even necessary. The absence of the tools limits the ability to validate estimated network structures and complicates the differentiation between common and unique effects in scientific applications.

\subsection{Related Works}
\label{sse:related}

There are several other VAR-type models that explicitly identify common and unique components. \citet{wilms2018multiclass} developed a multiclass VAR model in which the vectorized VAR transition matrices are assumed to be similar across classes, corresponding to subjects in our setting. Their method employs $\ell_1$ and fused Lasso penalties to enforce sparsity and similarity across subjects. However, their approach focuses on recovering similarity in individual dynamics across groups rather than having common and unique components.

A model introduced by \citet{skripnikov2019joint} is more closely aligned with our setting as it explicitly distinguishes between common and unique components. Specifically, the authors assume that while the supports of the common components are shared across subjects, their exact values may differ. This decomposition is achieved by incorporating a nonconvex penalty to ensure that the support sets remain nonoverlapping. Similarly, \citet{manomaisaowapak2022joint} proposed three variants of joint VAR models. In their framework, a weighted group Lasso penalty is applied to identify common paths across lags; a combination of weighted group Lasso penalties is implemented to identify both common and individual paths across lags; and a combination of a weighted group Lasso and a fused Lasso penalty is utilized to encourage similarity among individual paths across subjects. However, a method for a clear distinction between the unique components and the common paths is not explicitly provided.

More recently, a study by \cite{lyu2024high} focuses on covariate-driven population patterns with latent VAR formulations rather than heterogeneous subject-specific effects. In this framework, the observations are generated by a latent VAR model whose transition matrix is decomposed into a low-dimensional individual covariate multiplied by a common sparse matrix, plus a random measurement error component. The primary aim is to identify population-level patterns explained by covariates, whereas our focus using the multi-VAR model is explicitly designed to capture both shared and subject-unique dynamics across subjects, with an emphasis on establishing identifiability and providing formal inferential tools. Moreover, their formulation is conceptually closer to low-rank VAR models \cite[e.g.,][]{basu2019low}.

\subsection{Contributions}

This work makes several main contributions. First, we propose an identifiability condition for the multi-VAR model, called multi-VAR with structural expansion (multi-VARSE), grounded in the communication-efficient data integration framework of \citet{maity2022meta}. This condition provides a statistically principled alternative to median-based identification, enhancing estimation efficiency and ensuring robustness in heterogeneous settings. Second, building on this foundation, we develop a hypothesis testing framework specifically designed for high-dimensional VAR models across multiple subjects. Our framework enables rigorous assessment of the significance of common paths, as well as the homogeneity and nullity of individual paths across subjects. We can also control the false discovery rate (FDR) with a guarantee of statistical power.  Third, numerical studies demonstrate that multi-VARSE improves estimation performance across different levels of heterogeneity, while the proposed testing procedures achieve favorable FDR control and power. In the neuroimaging application, the framework yields more interpretable network structures than existing multi-VAR approaches through substantially sparser recovery.

\subsection{Organization of Paper}
\label{sse:organization}

The rest of the paper is organized as follows. Section \ref{se:model} defines the multi-VAR model and introduces our alternative identifiability condition. Sections \ref{se:estimation} and \ref{se:tests} present the estimation and inference frameworks, respectively. Section \ref{se:illustrative} reports numerical experiments comparing our method with the previous approaches and evaluates the hypothesis testing results. Section \ref{se:data} applies the proposed methods to empirical neuroimaging data. Finally, Section \ref{se:discuss} concludes. The statements of the theoretical results, their corresponding proofs, and additional simulation results are provided in the Supplementary Material.

\section{Model Description}
\label{se:model}

\subsection{Multi-VAR Model}\label{sse:previous_multiVAR}

Consider a $d$-dimensional observation series $\{X_t^{(k)}\}$ from $k\in[K]$ subjects observed at time points $t=1,\ldots,T_k$. The variables are assumed to be common across subjects, while the sample lengths may vary. We assume that each series follows a stationary VAR$(p)$ model,
\begin{align}\label{e:var}
    X_{t}^{(k)}
    =
    \Phi_1^{(k)} X_{t-1}^{(k)}
    + \ldots
    + \Phi_p^{(k)} X_{t-p}^{(k)}
    + \epsilon_{t}^{(k)},
    \quad
    \epsilon_{t}^{(k)}
    \sim
    \mathrm{WN}\left(
    0,
    \Sigma_{\epsilon}^{(k)}
    =
    \textrm{diag}
    (\sigma_{k,1}^{2},\ldots,\sigma_{k,d}^{2})
    \right).
\end{align}
The lag order $p\geq1$ is assumed to be the same across subjects. Let $N_k=T_k-p$, and define the stacked response matrix $\mathcal{Y}^{(k)} = (X_{p+1}^{(k)} \ \cdots \ X_{T_k}^{(k)})' \in \mathbb{R}^{N_k\times d}$ and the lagged design matrix $\mathcal{X}^{(k)}\in\mathbb{R}^{N_k\times dp}$ whose $t^{\textrm{th}}$ row is $\mathcal{X}_{t:}^{(k)} = (X_{p+t-1}^{(k)'} \ \ldots \ X_t^{(k)'})$, $t=1,\ldots,N_k$. Then the model \eqref{e:var} can be written as $\mathcal{Y}^{(k)} =\mathcal{X}^{(k)}B^{(k)} + E^{(k)}$, where $B^{(k)} =
((\Phi_1^{(k)})' \ \ldots \ (\Phi_p^{(k)})')'\in\mathbb{R}^{dp\times d}$ and $E^{(k)}\in\mathbb{R}^{N_k\times d}$ is the corresponding innovations matrix.

The multi-VAR model employs a joint estimation framework to obtain the decomposition \eqref{e:decomposition} of the VAR transition matrices in \eqref{e:var} across all $K$ subjects. Specifically, it builds on the data-shared Lasso \citep{gross2016data} or the stratified Lasso \citep{ollier2017regression}; $(\hat{\alpha}^{(0)},\hat{\beta}^{(1)},\ldots,\hat{\beta}^{(K)})$ is the solution of 
\begin{equation}\label{e:stratified_lasso}
    \argmin_{(\alpha^{(0)},\beta^{(1)},\ldots,\beta^{(K)})\in\mathbb{R}^{(K+1)d^2p}} \left\{\sum_{k=1}^K \frac{1}{2N_k}\|\textrm{vec}(\mathcal{Y}^{(k)}) - (I_d \otimes \mathcal{X}^{(k)}) \beta^{(k)}\|_2^2  + \tilde{\lambda}_0\|\alpha^{(0)}\|_1 + \sum_{k=1}^K \tilde{\lambda}_k \|\beta^{(k)} - \alpha^{(0)}\|_1\right\},
\end{equation}
where $\tilde{\lambda}_k>0$, $k\in[\bar{K}]:=\{0,1,\ldots,K\}$, are tuning parameters for the regularization, and $\|\cdot\|_r$ is the $\ell_r$ norm. The numerical implementation uses the fast iterative shrinkage-thresholding algorithm (FISTA; \cite{beck2009fast}) by stacking all individual equations. Specifically, let $\boldsymbol{Y} = (\boldsymbol{Y}^{(1)'},\ldots,\boldsymbol{Y}^{(K)'})'$ where $\boldsymbol{Y}^{(k)}=\textrm{vec}(\mathcal{Y}^{(k)})$ and similarly $\mathbf{E}^{(k)}$ in terms of $E^{(k)}$, $\boldsymbol{Z}^{(k)} = I_d \otimes \mathcal{X}^{(k)}$, and $\boldsymbol{\theta} = (\alpha^{(0)'},\alpha^{(1)'},\ldots,\alpha^{(K)'})'$. Then the aggregated regression is $\boldsymbol{Y} = \boldsymbol{Z}\boldsymbol{\theta} + \boldsymbol{E}$, where the $k^{\textrm{th}}$ block row of $\boldsymbol{Z}$ is
\begin{displaymath}
    \left[ \underbrace{\boldsymbol{Z}^{(k)}}_{\text{1st block column}}, \,\boldsymbol{0},\ldots,\boldsymbol{0},\, \underbrace{\boldsymbol{Z}^{(k)}}_{\text{$(k+1)$th block column}}, \,\boldsymbol{0},\ldots,\boldsymbol{0} \right].
\end{displaymath}
Then, if $N_k=N$ for all $k\in[K]$ for simplicity, the optimizer solves
\begin{displaymath}
    \hat{\boldsymbol{\theta}} = \argmin_{\boldsymbol{\theta} \in \mathbb{R}^{(K+1)d^2p}}\frac{1}{2N}\|\boldsymbol{Y} - \boldsymbol{Z}\boldsymbol{\theta}\|_2^2 + \tilde{\lambda}_0\|\boldsymbol{\theta}\|_{1,w},
\end{displaymath}
where $\|\cdot\|_{1,w}$ refers to the weighted $\ell_1$ norm with respect to $(K+1)d^2p$-dimensional nonnegative weights $\left(1_{d^2p}',\tilde{\lambda}_{1}/\tilde{\lambda}_{0}1_{d^2p}',\ldots,\tilde{\lambda}_{K}/\tilde{\lambda}_{0}1_{d^2p}'\right)' =  ((1,\tilde{\lambda}_{1}/\tilde{\lambda}_0,\ldots,\tilde{\lambda}_{K}/\tilde{\lambda}_0) \otimes 1_{d^2p}')'$, and $1_{d^2p}$ is a $d^2p$-dimensional vector consisting of 1's. Additional computational strategies, such as the backtracking step-size rule, are also included. The multi-VAR modeling and its estimation algorithm are implemented in the R package \texttt{multivar} \citep{fisher2021multivar}.

Similar to the standard Lasso estimation, uniformly applied tuning parameters can bias the multi-VAR estimators; an effective alternative within the multi-VAR framework is using the adaptive Lasso \citep{zou2006adaptive}. Specifically, the estimation problem is formulated as
\begin{equation}\label{e:stratified_adaptive_lasso}
    \argmin_{(\alpha^{(0)},\beta^{(1)},\ldots,\beta^{(K)})\in\mathbb{R}^{(K+1)d^2p}} \left\{\sum_{k=1}^K \frac{1}{2N_k}\|\textrm{vec}(\mathcal{Y}^{(k)}) - (I_d \otimes \mathcal{X}^{(k)}) \beta^{(k)}\|_2^2  + \tilde{\lambda}_0\|\alpha^{(0)}\|_{1,w} + \sum_{k=1}^K \tilde{\lambda}_k \|\beta^{(k)} - \alpha^{(0)}\|_{1,w}\right\}, 
\end{equation}
where the weights of $\|\cdot\|_{1,w}$ in \eqref{e:stratified_adaptive_lasso} are constructed by initially fitting individual VAR models with Lasso to obtain $\{\hat{\beta}^{(k)}\}$, and then setting the weights for the common path as $w_i^{(0)}=1/|\hat{\beta}_i^{(0)}|$, where $\hat{\beta}_i^{(0)}=\textrm{median}(\hat{\beta}_i^{(1)},\ldots,\hat{\beta}_i^{(K)})$, and the weights for the unique paths as $w_i^{(k)}=1/|\hat{\beta}_i^{(0)}-\hat{\beta}_i^{(k)}|$, $k\in[K]$, and $i\in[d^2p]$. A similar equation, stacked across all subjects, is applied in this adaptive weighting scheme, which is the default setting of \texttt{multivar}.

In the ordinary regression setting, \cite{asiaee2019data} pointed out that the estimators $\{\hat{\alpha}^{(k)}\}_{k\in[\bar{K}]}$ in \eqref{e:stratified_lasso} suffer from statistical inefficiency, exhibiting a slower convergence rate than that of individual paths estimated via separate VAR models. That is, by taking $N_0=\sum_{k=1}^K N_k$,
\begin{displaymath}
    \|\hat{\alpha}^{(0)} - \alpha^{(0)}\|_2 + \sum_{k=1}^{K}\sqrt{\frac{N_k}{N_0}}\|\hat{\alpha}^{(k)} - \alpha^{(k)}\|_2 
    = \max_{k\in[K]}\frac{N_0}{N_k} \cdot \mathcal{O}_{\mathbb{P}}\left( \sqrt{\frac{\max_{k\in[\bar{K}]} (\|\alpha^{(k)}\|_0) \log d^2p}{N_0}}\right).
\end{displaymath}
where $\mathcal{O}_{\mathbb{P}}(\cdot)$ denotes a bound that holds with high probability as $\log d /N_0 \to 0$. By letting $N_k=N$, $k\in[K]$, one can expect that the convergence rates of both common and unique path estimators are proportional to $\sqrt{\max_{k\in[\bar{K}]}(\|\alpha^{(k)}\|_0)\log d^2p/N}$, rendering the overall rate dependent on the least sparse subject. This dependency can be particularly problematic when applying multi-VAR modeling to ILD. Such datasets are typically heterogeneous and frequently feature short sample lengths relative to their high dimensionality, while gaining efficiency from the samples of different subjects is not possible.

\subsection{Proposed Identification}

Following the robust estimation perspective of \cite{maity2022meta}, identifying common paths from a collection of individual paths can be viewed as analogous to robustly estimating a location parameter from potentially heterogeneous observations. Specifically, for each coordinate $(i,j)$, $i\in[d]$, and $j\in[dp]$, the collection $\{(\beta_i^{(k)})_j\}_{k=1}^K$ may consist of values concentrated around a common location together with subject-specific deviations that exhibit substantial heterogeneity. From this perspective, our objective is to define the common path at each coordinate as a location parameter that captures the structure shared across individual paths while limiting the influence of highly heterogeneous unique paths, which can be regarded as outliers.

As a substitution for the sum of $\ell_1$ penalties, the conventional approach to identifying common paths relies on mean characterization, which substitutes the absolute values in \eqref{e:stratified_adaptive_lasso} with squared $\ell_2$ penalties. Although this formulation can improve the statistical efficiency of the common components by leveraging samples across all subjects, the estimators become sensitive to outliers. Consequently, the estimated common paths tend to be substantially influenced by subject-level heterogeneity. This might also hinder model interpretability, as even a few highly heterogeneous outliers can easily pull the common paths away from zero.

We define the common components using a redescending loss function, which limits the influence of extreme subject-unique heterogeneity by forcing its derivative to vanish outside a bounded interval \cite[e.g., Chapter 4.8 in][]{huber2011robust}. This property effectively treats highly outlying unique parameters as outliers, ignoring them to preserve both the interpretability and sparsity of the common structure. A properly determined threshold allows us to maintain statistical efficiency by borrowing strength from the samples that contribute to the common paths. Consequently, the common paths in our proposed framework are defined element-wise as
\begin{displaymath}
    (\alpha_{i}^{(0)})_j 
    = \argmin_{x \in \mathbb{R}} \left\{ \sum_{k=1}^K \Psi_{\eta_{i,j}} \left( (\beta_{i}^{(k)})_j - x \right) \right\}, \quad i\in[d], \ j\in[dp],
\end{displaymath}
where $\Psi_{\eta}(x) = \min\{x^2,\eta^2\}$ is a truncated quadratic loss parameterized by the tuning threshold $\eta>0$, whose estimator analog is detailed in Section \ref{sse:approach_estimation}.

\section{Estimation}
\label{se:estimation}

\subsection{Estimation Procedure}
\label{sse:approach_estimation}

The key difference from the original multi-VAR estimation framework is that we first estimate the individual VAR models separately and then aggregate them, rather than estimating all parameters jointly as in \eqref{e:stratified_lasso} and \eqref{e:stratified_adaptive_lasso}.
    
\noindent \textbf{Sparse estimation:} First, we set up equations for each variable separately for the VAR models of the $k^{\textrm{th}}$ subject, $k\in[K]$. Note that the $i^{\textrm{th}}$ equation in the $d$-dimensional VAR$(p)$ model is written as
\begin{displaymath}
    X_{i,t}^{(k)} 
    = \sum_{\ell=1}^p [\Phi_{\ell}^{(k)}]_{i:}X_{t-\ell}^{(k)} + \epsilon_{i,t}^{(k)} 
    =: \mathcal{X}_{t:}^{(k)}\beta_{i}^{(k)} + \epsilon_{i,t}^{(k)},
\end{displaymath}
where $[M]_{i:}$ is the $i^{\textrm{th}}$ row of the matrix $M$. The regression formulation for the $i^{\textrm{th}}$ variable in a matrix form is $\mathcal{Y}_i^{(k)} = \mathcal{X}^{(k)} \beta_{i}^{(k)} + E_i^{(k)}$, where $\mathcal{Y}_i^{(k)} = (X_{i,p+1}^{(k)},\ldots,X_{i,T_k}^{(k)})'$, $E_i^{(k)} = (\epsilon_{i,p+1}^{(k)},\ldots,\epsilon_{i,T_k}^{(k)})'$, and the design matrix $\mathcal{X}^{(k)}$, defined in Section \ref{sse:previous_multiVAR}. For each $i^{\textrm{th}}$ coordinate, we get the estimator $\hat{\beta}_i^{(k)}\in\mathbb{R}^{dp}$ by the Lasso method,
\begin{equation}\label{e:lasso}
    \hat{\beta}_i^{(k)} = \argmin_{\beta_i^{(k)} \in \mathbb{R}^{dp}}\left\{ \frac{1}{2N_k} \left\|\mathcal{Y}_i^{(k)} -  \mathcal{X}^{(k)}\beta_i^{(k)} \right\|_2^2 + \lambda_i^{(k)}\|\beta_i^{(k)}\|_1\right\},
\end{equation}
where $\lambda_i^{(k)}>0$ is a tuning parameter.

\noindent \textbf{Debiasing:} By following the debiased Lasso estimation \citep[e.g.,][]{basu2024high,adamek2023lasso}, one can obtain a debiased $dp$-dimensional estimator for the $i^{\textrm{th}}$ equation using \eqref{e:lasso} as 
\begin{equation}\label{e:debiased}
    \tilde{\beta}_{i}^{(k)} = \hat{\beta}_{i}^{(k)} + \frac{1}{N_k}\hat{\Theta}^{(k)}\mathcal{X}^{(k)'}(\mathcal{Y}_i^{(k)} - \mathcal{X}^{(k)}\hat{\beta}_{i}^{(k)}), \quad i\in[d],
\end{equation}
where $\hat{\Theta}^{(k)}$ is the approximate inverse of the sample covariance of the residual $\mathcal{Y}_i^{(k)} - \mathcal{X}^{(k)}\hat{\beta}_{i}^{(k)}$ regarding  $k^{\textrm{th}}$ subject, which is common across all $\tilde{\beta}_{i}^{(k)}$, $i\in[d]$. It is defined as $\hat{\Theta}^{(k)} = (\hat{\gamma}^{(k)})^{-2}\hat{\Gamma}^{(k)}$, where 
\begin{displaymath}
    \hat{\Gamma}^{(k)} = \begin{pmatrix}
        1 & -\hat{\gamma}_{1,2}^{(k)} & \ldots & -\hat{\gamma}_{1,dp}^{(k)} \\
        -\hat{\gamma}_{2,1}^{(k)} & 1 & \ldots & -\hat{\gamma}_{2,dp}^{(k)} \\
        \vdots & \vdots & \ddots & \vdots \\
        -\hat{\gamma}_{dp,1}^{(k)} & -\hat{\gamma}_{dp,2}^{(k)} & \ldots & 1
    \end{pmatrix}
\end{displaymath}
and each $(dp-1)$-dimensional vector $\hat{\gamma}_j^{(k)} = \{\hat{\gamma}_{j,l}^{(k)},l \in [dp]\setminus\{j\}\}$ is obtained by the nodewise regression,
\begin{equation}\label{e:nodewise}
    \hat{\gamma}_j^{(k)} = \argmin_{\gamma_j^{(k)}\in\mathbb{R}^{dp-1}}\left\{\frac{1}{N_k} \left\|\mathcal{X}_j^{(k)} - \mathcal{X}_{-j}^{(k)}\gamma_{j}^{(k)} \right\|_2^2 + \lambda_{j}^{(k)}\|\gamma_j^{(k)}\|_1\right\},
\end{equation}
where $\mathcal{X}_j^{(k)}$ is the $j^{\textrm{th}}$ column of $\hat{\Gamma}^{(k)}$ and $\mathcal{X}_{-j}^{(k)}$ is $\hat{\Gamma}^{(k)}$ with the $j^{\textrm{th}}$ column removed. By taking $(\hat{\tau}_{j}^{(k)})^2 = \frac{1}{N_k}\|\mathcal{X}_j^{(k)} - \mathcal{X}_{-j}^{(k)}\hat{\gamma}_{j}^{(k)} \|_2^2 + \lambda_{j}^{(k)}\|\hat{\gamma}_j^{(k)}\|_1$, we set
\begin{displaymath}
    (\hat{\gamma}^{(k)})^{-2} = \textrm{diag}(1/(\hat{\tau}_1^{(k)})^2,\ldots,1/(\hat{\tau}_{dp}^{(k)})^2).
\end{displaymath}

\noindent \textbf{Identification:} Finally, we obtain the estimator of the common paths by aggregating individual estimators. To obtain the estimator of $(\alpha_{i}^{(0)})_j$, we minimize the sum of the redescending loss function
\begin{equation}\label{e:redescending_estim}
    (\tilde{\alpha}_{i}^{(0)})_j
    = \argmin_{(\alpha_{i}^{(0)})_j\in\mathbb{R}}\left\{L_{i,j}((\alpha_{i}^{(0)})_j) 
    := \sum_{k=1}^K\Psi_{\eta_{i,j}}((\tilde{\beta}_{i}^{(k)})_j-(\alpha_{i}^{(0)})_j) = \sum_{k=1}^K \min\{((\tilde{\beta}_{i}^{(k)})_j-(\alpha_{i}^{(0)})_j)^2,\eta_{i,j}^2\}\right\},
\end{equation}
where $\{\eta_{i,j}\}$ are tuning parameters of the redescending loss function. Naturally, the unique paths are defined as $(\tilde{\alpha}_{i}^{(k)})_j = (\tilde{\beta}_{i}^{(k)})_j - (\tilde{\alpha}_{i}^{(0)})_j$. To recover the sparsity, the hard thresholding is applied to $(\tilde{\alpha}_{i}^{(0)})_j$ and $(\tilde{\alpha}_{i}^{(k)})_j$ to produce sparse estimators $(\hat{\alpha}_{i}^{(0)})_j$ and $(\hat{\alpha}_{i}^{(k)})_j$, $k\in[K]$, respectively, $j\in[dp]$ and $i\in[d]$. The threshold functions are defined as
\begin{align}
    HT_{\delta_k}(\theta) &= \theta1_{\{|\theta|\geq \delta_k\}}, \qquad \theta\in\mathbb{R}. \label{e:hard_thres} 
\end{align}
The ideal values of the threshold parameters are $\delta_k \propto \sqrt{\frac{\log dp}{N_k}}$, $k\in[K]$, and $\delta_0 \propto \sqrt{\frac{\log dp}{K N_{\min}}}$, where $N_{\min}= \min_{k\in[K]} N_k$, which are used for theoretical statements in the Supplementary Material.

The convergence rates of the estimated common and unique paths to their population counterparts are stated in Proposition \ref{prop:estim} in the Supplementary Material. Unlike the existing multi-VAR estimator whose rate is governed by the least sparse subject (Section \ref{sse:previous_multiVAR}), the proposed framework achieves component-wise rates: $\mathcal{O}_{\mathbb{P}}\left(\sqrt{\|\alpha^{(0)}\|_0\log d^2 /(KN_{\min})}\right)$ for the common path, and $\mathcal{O}_{\mathbb{P}}\left(\sqrt{\|\alpha^{(k)}\|_0\log d^2/N_k}\right)$ for each unique path $k\in[K]$, with each rate depending only on its own sparsity level.

\begin{remark}\label{rem:cv_rule}
To select the tuning parameters $\{\delta_k\}_{k=0}^K$, and $\{\eta_{ij}\}_{i,j}$ in practice, we implement a three-layer cross-validation procedure. Based on empirical experiments, we set the search grids as follows: $c_0$ ranges from 0.1 to 1 at equal intervals, $c_K$ ranges from 0.5 to 1 at equal intervals, and $\eta=\eta_{i,j}$ ranges evenly between $\min_{i,j,k} (\tilde{\beta}_{i}^{(k)})_j$ and $\max_{i,j,k} (\tilde{\beta}_{i}^{(k)})_j$. These grids control the constants $c_0$ and $c_K$, which define our empirical thresholding parameters: for the condition numbers of the estimated residual covariance matrices $\{\kappa(\hat{\Sigma}_{\varepsilon}^{(k)})\}_{k=1,\ldots,K}$,
\begin{displaymath}
    \delta_0 
    = \max_{k\in[K]}\kappa(\hat{\Sigma}_{\varepsilon}^{(k)})\sqrt{\frac{\log dp}{c_0K N_{\min}}}, \quad
    \delta_k = c_K\kappa(\hat{\Sigma}_{\varepsilon}^{(k)})\sqrt{\frac{\log dp}{N_k}}, \ k\in[K].
\end{displaymath}
While the structure of $\delta_0$ aligns with the theoretical framework in Proposition \ref{prop:estim} in the Supplementary Material, the formulation of $\delta_k$ used in this cross-validation is determined empirically. This joint, grid-based approach is necessitated by several practical limitations that make ideal, individualized tuning unfeasible. First, finding the parameters $\eta_{i,j}$ and $\delta_{k}$ for every single element and subject individually is computationally intractable. Second, from our empirical findings, standard cross-validation metrics, such as the mean prediction error averaged across all $K$ subjects, are highly sensitive to changes in $\eta_{i,j}$ but largely insensitive to variations in the grids of the other thresholding constants. Compounding this difficulty, we observed that the final model sparsity is highly dependent on the exact choices of $c_0$ and $c_K$. In certain scenarios, the thresholding mechanism for unique components can become overly conservative; without careful scaling, it tends to eliminate unique paths, even when the underlying cross-validation errors remain statistically indistinguishable. 
\end{remark}

\section{Hypothesis Tests}
\label{se:tests}

\subsection{Single Test Procedure}
\label{sse:test_single}

Let
\begin{displaymath}
    \hat{V}_{i,j}^{(k)} 
    := \hat{\sigma}_{k,i}^2\hat{\Omega}_{j,j}^{(k)}
    := \hat{\sigma}_{k,i}^2[\hat{\Theta}^{(k)}\hat{\Sigma}^{(k)}\hat{\Theta}^{(k)'}]_{j,j}, \quad i\in[d],\ j\in[dp], \ k\in[K],
\end{displaymath}
where $\hat{\Sigma}^{(k)}=\frac{1}{N_k}\mathcal{X}^{(k)'}\mathcal{X}^{(k)}$, $\hat{\Theta}^{(k)}$ is defined in \eqref{e:debiased}, and 
\begin{equation}\label{e:def_sigma}
    \hat{\sigma}_{k,i}^2 = \frac{1}{N_k}\sum_{t=1}^{N_k}\left(X_{i,t}^{(k)}  - \mathcal{X}_{t:}^{(k)}\hat{\beta}_{i}^{(k)}\right)^2, \quad k\in[K], \ i\in[d].
\end{equation}
%

\noindent \textbf{Test of significance:} First, we introduce a hypothesis test that assesses the significance of a common path $(\alpha_{i}^{(0)})_j$ across the subjects. Rejection of the corresponding null hypothesis indicates that a joint analysis may be advantageous. Note that this workflow differs from the standard ANOVA analysis, which typically only focuses on examining heterogeneity across subjects. We intentionally begin with the significance test, as failing to reject the null suggests little benefit in pursuing a joint analysis with the multi-VAR model.

Suppose that, for $i\in[d]$ and $j\in[dp]$, we are interested in
\begin{equation}\label{e:h0_common}
    H_{0,(i,j)}^{(0)}: (\alpha_{i}^{(0)})_j = 0  \quad vs \quad 
    H_{1,(i,j)}^{(0)}: (\alpha_{i}^{(0)})_j \neq 0.
\end{equation}
For fixed $\eta_{i,j}$, let 
\begin{displaymath}
    \mathcal{J}_{(i,j)} := \left\{ k\in[K] : |(\tilde{\beta}_{i}^{(k)})_j - (\tilde{\alpha}_{i}^{(0)})_j|\leq \eta_{i,j} \right\}
\end{displaymath}
be the set of subject indices that contribute to the common path $(i,j)$. Define $N_{(i,j)}=\frac{1}{|\mathcal{J}_{(i,j)}|}\sum_{k\in\mathcal{J}_{(i,j)}}N_k$. Proposition \ref{prop:test_single}(a) in the Supplementary Material establishes that under $H_{0,(i,j)}^{(0)}$,
\begin{equation}\label{e:clt_alpha0}
    Z_{(i,j)}
    := \frac{\sqrt{N_{(i,j)}}(\tilde{\alpha}_{i}^{(0)})_j}{\sqrt{\frac{1}{|\mathcal{J}_{(i,j)}|^2}\sum_{k\in\mathcal{J}_{(i,j)}}\hat{V}_{i,j}^{(k)}}} \stackrel{d}{\to} \mathcal{N}(0,1)
\end{equation}
as $N_{(i,j)}\to\infty$. We reject the null hypothesis if $|Z_{(i,j)}| \geq z_{1-\alpha/2}$.

\noindent \textbf{Test of homogeneity:} Next, we test whether the individual paths are homogeneous. For $i\in[d]$ and $j\in[dp]$, consider
\begin{displaymath}
    H_{0,(i,j)}^{[K]}: (\alpha_{i}^{(1)})_j = \ldots = (\alpha_{i}^{(K)})_j = 0 
    \quad vs \quad 
    H_{1,(i,j)}^{[K]}: \textrm{not }H_{0,(i,j)}^{[K]}.
\end{displaymath}
Under the decomposition \eqref{e:decomposition} and the proposed identification, the hypotheses are equivalent to
\begin{equation}\label{e:h0_homo}
    H_{0,(i,j)}^{[K]}: (\beta_{i}^{(1)})_j = \ldots = (\beta_{i}^{(K)})_j \quad vs \quad 
    H_{1,(i,j)}^{[K]}: \textrm{not }H_{0,(i,j)}^{[K]}.
\end{equation}
%

Since the individual paths are directly estimable through the debiased estimators, we construct the test from \eqref{e:h0_homo}. This formulation is equivalent to testing all pairwise equalities 
$(\beta_{i}^{(k)})_j=(\beta_{i}^{(\ell)})_j$ for $1\leq k<\ell\leq K$ simultaneously. Stacking all $K(K-1)/2$ pairwise statistics as 
$\left(Z_{(i,j)}^{(k,\ell)}\right)_{1\leq k<\ell\leq K}$, we form the squared $\ell_2$ norm
\begin{displaymath}
    \chi^2_{(i,j)} 
    = \sum_{1\leq k<\ell\leq K}\left(Z_{(i,j)}^{(k,\ell)}\right)^2,
\end{displaymath}
where $Z_{(i,j)}^{(k,\ell)}$ is 
\begin{displaymath}
    Z_{(i,j)}^{(k,\ell)} 
    = \frac{(\tilde{\beta}_{i}^{(k)})_j-(\tilde{\beta}_{i}^{(\ell)})_j}{\sqrt{\frac{1}{N_{k}}\hat{V}_{i,j}^{(k)} + \frac{1}{N_{\ell}}\hat{V}_{i,j}^{(\ell)}}}.
\end{displaymath}
Note that $\chi^2_{(i,j)}$ does not follow a standard chi-square distribution as the pairwise statistics $Z_{(i,j)}^{(k,\ell)}(0)$ are mutually correlated due to shared subjects. Following \cite{liu2017structural}, $\chi^2_{(i,j)}$ is known to asymptotically follows a weighted chi-square distribution
\begin{displaymath}
    \chi^2_{(i,j)} \stackrel{d}{\to} \sum_{m=1}^{M}\nu_m Z_m^2, 
    \quad Z_1,\ldots,Z_M\stackrel{i.i.d.}{\sim}\mathcal{N}(0,1),
\end{displaymath}
as $N_k \to \infty$ for all $k$, where $\nu_1\geq\cdots\geq\nu_M>0$ are the eigenvalues of $\Sigma_{(i,j)}$ with $M = \mathrm{rank}(\Sigma_{(i,j)})$, whose structure is described in the Supplementary Material, obtained through eigendecomposing its estimator $\hat{\Sigma}_{(i,j)}$ in practice. We convert $\chi^2_{(i,j)}$ into a scalar $z$-value via
\begin{equation}\label{e:z_homo}
    Z_{(i,j)}
    = \Phi^{-1}\left(\mathcal{F}_0\left(\chi^2_{(i,j)}\right)\right),
\end{equation}
where $\Phi(\cdot)$ is a CDF of standard normal distribution and $\mathcal{F}_0(x)=\mathbb{P}\left(\sum_{m=1}^M\nu_m Z_m^2 \leq x\right)$ is evaluated via Monte Carlo simulation. Note that, by using the corresponding $Z$-test statistics, the FDR for all three tests described in this section can be controlled using the same thresholding procedure (see Section \ref{sse:test_multiple}). Proposition \ref{prop:test_single}(b) in the Supplementary Material yields the following result: Under $H_{0,(i,j)}^{[K]}$, we reject the null hypothesis if $|Z_{(i,j)}|\geq z_{1-\alpha/2}$.

\noindent \textbf{Test of nullity:} As a derived case from the test of homogeneity, it assesses whether the individual paths are degenerate for all subjects. For $i\in[d]$ and $j\in[dp]$, consider
\begin{displaymath}
    H_{0,(i,j)}^{[K]_{\circ}} : (\alpha_{i}^{(0)})_j = (\alpha_{i}^{(1)})_j = \ldots 
    = (\alpha_{i}^{(K)})_j = 0 
    \quad vs \quad 
    H_{1,(i,j)}^{[K]_{\circ}}: \textrm{not }H_{0,(i,j)}^{[K]_{\circ}}.
\end{displaymath}
This null hypothesis is the intersection of the null hypotheses of the tests of significance and homogeneity. By the same argument as above, it is equivalent to
\begin{equation}\label{e:h0_null}
    H_{0,(i,j)}^{[K]_{\circ}} : (\beta_{i}^{(1)})_j = \ldots = (\beta_{i}^{(K)})_j=0 \quad vs \quad 
    H_{1,(i,j)}^{[K]_{\circ}}: \textrm{not }H_{0,(i,j)}^{[K]_{\circ}}.
\end{equation}
which amounts to testing $(\beta_{i}^{(k)})_j=0$ individually for all $k\in[K]$.

The test statistic
\begin{displaymath}
    \chi^2_{(i,j)}
    = \sum_{k=1}^{K} \frac{N_k((\tilde{\beta}_{i}^{(k)})_j)^2}{\hat{V}_{i,j}^{(k)}}
\end{displaymath}
asymptotically follows a standard chi-square distribution with $K$ degrees of freedom, since it is a sum of $K$ independent squared standard normal random variables under $H_{0,(i,j)}^{[K]_{\circ}}$. We convert $\chi^2_{(i,j)}$ into a scalar $z$-value via
\begin{equation}\label{e:z_null}
    Z_{(i,j)}
    = \Phi^{-1}\left(\mathcal{F}_{K}\left(\chi^2_{(i,j)}\right)\right).
\end{equation}
where $\mathcal{F}_K(x)$ is the CDF of the chi-square distribution with $K$ degrees of freedom. By Proposition \ref{prop:test_single}(c) in the Supplementary Material, we have the following: Under $H_{0,(i,j)}^{[K]_{\circ}}$, $Z_{(i,j)}$ asymptotically follows a standard normal distribution as $N_k \to \infty$ for all $k$, so we reject the null hypothesis if $|Z_{(i,j)}|\geq z_{1-\alpha/2}$.


\subsection{Multiple Test Procedure}
\label{sse:test_multiple}

In this section, we introduce a procedure for controlling the FDR of the proposed hypothesis tests across the entries in the set of interest $(i,j)\in \mathcal{S}\subseteq[d] \times [dp]$ among the coefficients of VAR transition matrices over all subjects.

The advantage of the proposed single test procedure is that a common thresholding procedure applies to all three tests. This is made possible by the probability integral transform constructions in \eqref{e:clt_alpha0}, \eqref{e:z_homo}, and \eqref{e:z_null}, which map the respective test statistics to a common pivot statistic following the standard normal distribution. As a result, the null distributions are asymptotically identical, and the two-sided rejection rule with the Benjamini-Hochberg-type threshold \citep{benjamini1995controlling} as the null tail probability applies. This uniformity is particularly convenient in practice, as one can simultaneously conduct multiple types of multiple tests without recalibrating the thresholds.

We define the null and alternative hypotheses for each multiple test.
\begin{itemize}
    \item Test of significance: 
    \begin{equation}\label{e:h0_common_multi}
        H_{0}^{(0)}: (\alpha_{i}^{(0)})_j=0, \ (i,j)\in\mathcal{S}, \quad vs \quad 
        H_{1}^{(0)}: \textrm{not }H_{0}^{(0)},
    \end{equation}
    \item Test of homogeneity: 
    \begin{equation}\label{e:h0_homo_multi}
        H_{0}^{[K]}: (\beta_{i}^{(1)})_j=\ldots=(\beta_{i}^{(K)})_j, \ (i,j)\in\mathcal{S}, \quad vs \quad 
        H_{1}^{[K]}: \textrm{not }H_{0}^{[K]},
    \end{equation}
    \item Test of nullity: 
    \begin{equation}\label{e:h0_null_multi}
        H_{0}^{[K]_{\circ}}: (\beta_{i}^{(1)})_j=\ldots=(\beta_{i}^{(K)})_j=0, \ (i,j)\in\mathcal{S}, \quad vs \quad 
        H_{1}^{[K]_{\circ}}: \textrm{not }H_{0}^{[K]_{\circ}}.
    \end{equation}
\end{itemize}
For each type of test, we compute the threshold 
\begin{equation}\label{e:thres_sig} 
    \hat{\zeta} = \inf \left\{ 0\leq\zeta\leq \sqrt{2\log(|\mathcal{S}|)}: \frac{2\bar{\Phi}(\zeta)|\mathcal{S}|}{\max\{1, \sum_{(i,j)\in\mathcal{S}} 1_{\{|Z_{(i,j)}| \geq \zeta\}}\}} \leq \alpha \right\},
\end{equation}
where $\bar{\Phi}(\zeta) = 1-\Phi(\zeta)$. Then we obtain the set of discoveries
\begin{equation}\label{e:discovery_set} 
    \hat{\mathcal{S}}(\hat{\zeta}) = \{(i,j)\in\mathcal{S}:|Z_{(i,j)}|\geq\hat{\zeta}\}.
\end{equation}

Proposition \ref{prop:test_multiple} in the Supplementary Material establishes the asymptotic results of FDR control and power of the proposed multiple testing procedure.

\section{Simulation Study}
\label{se:illustrative}

\subsection{Simulation Setups}\label{sse:simulation_setup}

The proposed estimation framework (multi-VARSE) detailed in Section \ref{sse:approach_estimation} is compared against the multi-VAR estimation in \eqref{e:stratified_lasso} (multi-VAR) and its adaptive Lasso variant in \eqref{e:stratified_adaptive_lasso} (multi-VAR (A)); these two benchmarks are implemented using the \texttt{R} package \texttt{multivar} \citep{fisher2021multivar}. Existing methods outside the multi-VAR framework either fail to estimate identically defined common and unique paths separately or rely on low-rank modeling with individualized covariates (see Section \ref{sse:related}). Furthermore, some alternative approaches cannot be modified to accommodate our simulation setup because publicly accessible code is unavailable.

We focus on the case $p=1$ with independent Gaussian innovations $\epsilon_{i,t}^{(k)} \sim \mathcal{N}(0,0.5^2)$. Among the $d^2$ possible paths, $s_0 d^2$ are designated as common, and $\left(\sum_{k=1}^K s_k\right) d^2$ unique paths are selected such that no overlaps occur across subjects. For each selected path, we generate coefficients using $\textrm{Unif}(0.1,0.9)$. We repeat this generation until all spectral norms of $\Phi^{(k)}$, $k\in[K]$, are below 0.95. For estimation, we set $K \in \{10,15\}$, varying $d \in \{10,20\}$ and the average sample length $T \in \{100,200\}$, with individual sample lengths varying within $\pm 10$ of $T=100$ and $\pm 20$ of $T=200$. Three relative heterogeneity levels are considered: the low setting fixes $(s_0,s_k)=(0.04,0.02)$ for all $k\in[K]$; the medium setting fixes $(s_0,s_k)=(0.04,0.04)$ for all $k\in[K]$; and the high setting fixes $s_0=0.04$, where $\lfloor K/2 \rfloor$ subjects have $s_k=0.04$ and the remaining subjects have $s_k=0.08$. We repeat the simulations 50 times for each setup. As the computational costs of both multi-VAR and multi-VAR (A) escalate substantially with the dimension and number of subjects, these approaches cannot be evaluated in higher-dimensional settings. Simulation results for multi-VARSE under higher dimensions ($d=30,50$) and different setups \citep{fisher2024structured} are provided in Sections \ref{sse:additional_higher} and \ref{sse:additional_different} of the Supplementary Material, respectively.

For estimation, we compute the root mean square error (RMSE), sensitivity (Sens), and specificity (Spec) for ${\alpha^{(0)}}$ and the averaged ${\alpha^{(k)}}$ across the $K$ subjects. For hypothese tests, we compute the empirical FDR and power at a nominal significance level of $\alpha=0.05$ from the multi-VARSE estimation. Details on how performance metrics are computed are also described in Section \ref{sse:additional_setups} of the Supplementary Material.

\subsection{Estimation Results}\label{sse:result_estimation}

The simulation results for estimation accuracy are presented in Figure \ref{fig:sim_estimation}. Across both lower- and higher-dimensional settings, the proposed Multi-VARSE method consistently achieves lower or comparable RMSE for both common and unique paths relative to the alternative approaches. This performance gap is particularly prominent in high-heterogeneity settings, where multi-VARSE, followed by multi-VAR (A), maintains near-perfect specificity for recovering sparse entries, while multi-VAR suffers a sharp drop. As expected, increasing the sample size significantly reduces RMSE and narrows the gap between the methods. The individually adjusted thresholding applied in our framework appears to function similarly to adaptive weights in a penalization scheme, maintaining robust sensitivity and superior specificity across all conditions. Finally, varying the number of subjects yields no substantial differences in performance metrics, aligning with findings in previous multi-subject time series studies \cite[e.g.,][]{fisher2022penalized,kim2024group}.

\begin{figure}[th!]
 \centering
 \includegraphics[width=1\textwidth,height=0.8\textheight]
 {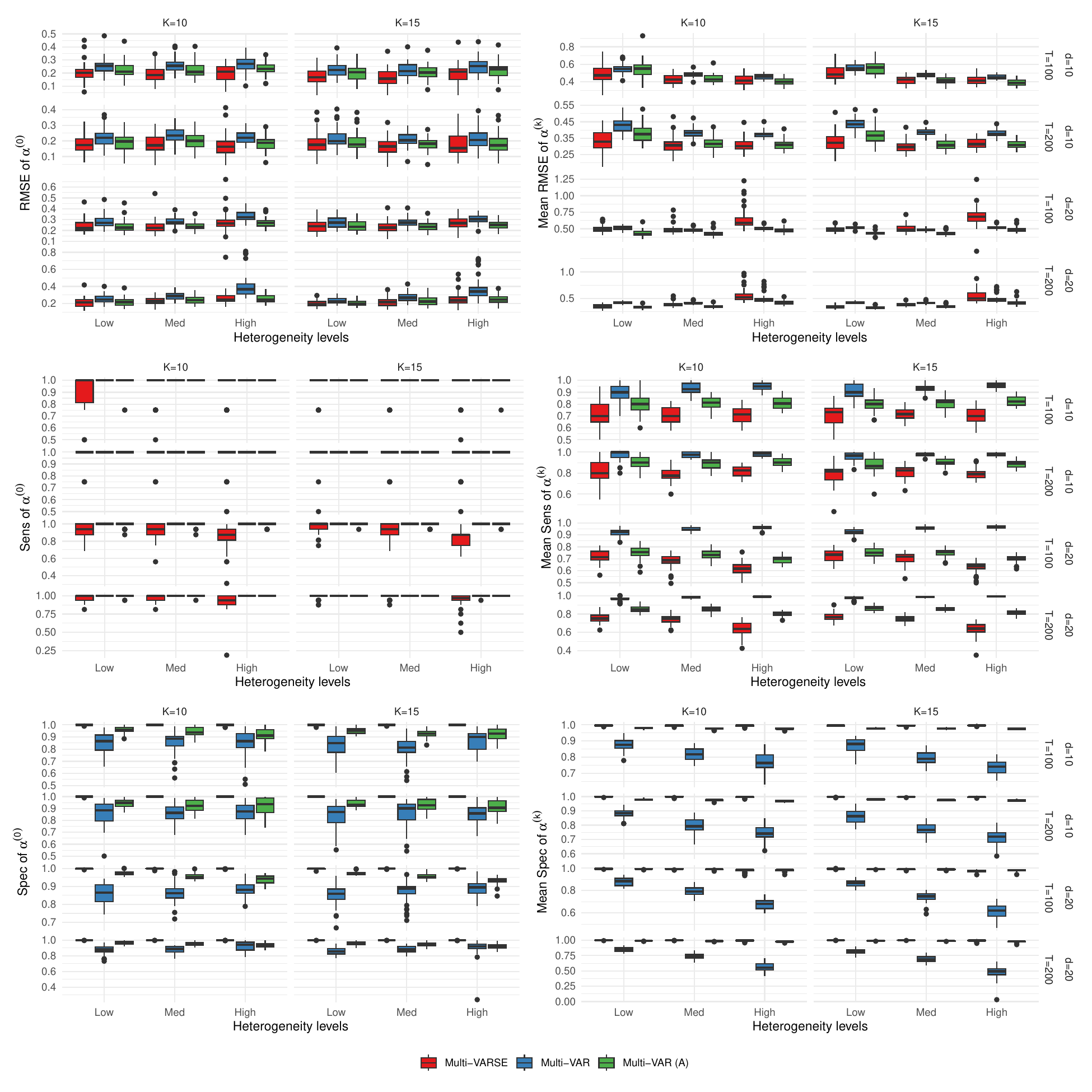}
 \caption{Boxplots of the root mean square error (RMSE) of $\alpha^{(0)}$ (top left), the average RMSE of $\alpha^{(k)}$ (top right), the sensitivity (Sens) of $\alpha^{(0)}$ (middle left), the average sensitivity of $\alpha^{(k)}$ (middle right), the specificity (Spec) of $\alpha^{(0)}$ (bottom left), and the average specificity of $\alpha^{(k)}$ (bottom right) under different combinations of $d$ and average $T$ (combinations indicated on the right tabs), $K$ (each column), and heterogeneity levels (each axis). Red indicates the proposed method, while blue and green represent the benchmark methods.}
 \label{fig:sim_estimation}
\end{figure}


\subsection{Hypothesis Tests Results}\label{sse:result_test}

The simulation results for the hypothesis tests are shown in Figure \ref{fig:sim_test}. Changing the number of subjects does not noticeably alter the results. For the test of significance, performance is nearly perfect: the FDR stays at zero, and the power is consistently high across all settings. In contrast, the tests of homogeneity and nullity are more sensitive to the number of variables and the level of heterogeneity. As the dimension increases, the FDR decreases, while the average power remains relatively low despite an increase in its maximum. Furthermore, as the level of heterogeneity increases from low to high, the FDR increases, while the power declines slightly. As expected, larger sample sizes generally improve performance by reducing the FDR and increasing power across all tests.

\begin{figure}[ht!]
 \centering
 \includegraphics[width=1\textwidth,height=0.8\textheight]
 {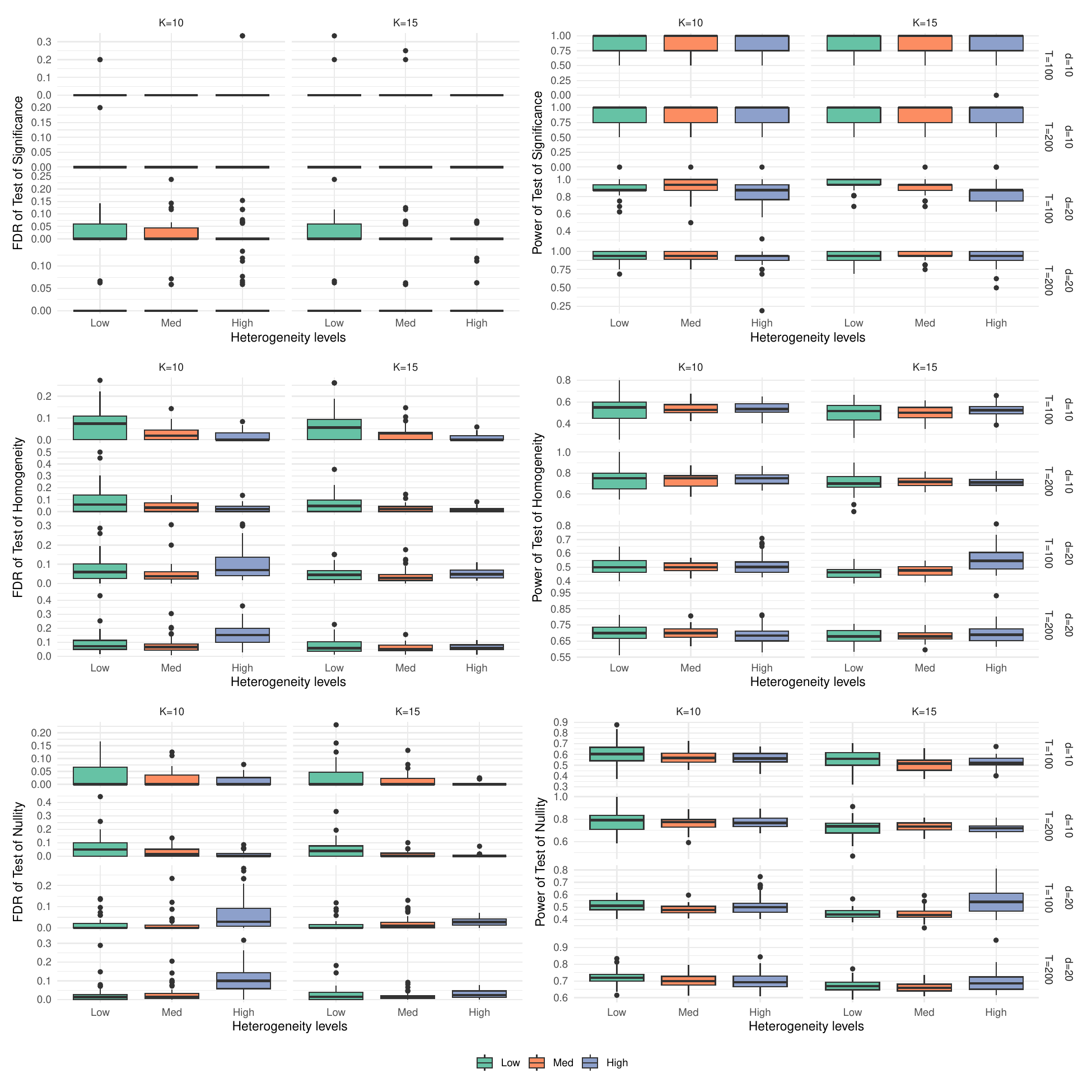}
 \caption{Boxplots of the FDRs (left columns) and powers (right columns) for the three hypothesis tests: test of significance (top), test of homogeneity (middle), and test of nullity (bottom), under different combinations of $d$  and average $T$ (combinations indicated on the right tabs), $K$ (each column), and heterogeneity levels (each axis), presented through three colors.}
 \label{fig:sim_test}
\end{figure}

\section{Data Application}
\label{se:data}

\subsection{Data Description}

We use task fMRI (tfMRI) data from the WU-Minn Human Connectome Project (HCP) \citep{van2013wu}. The data have been preprocessed through the minimal pipeline described in \cite{glasser2013minimal}. The HCP emotion processing task probes brain circuits involved in affective perception, particularly the amygdala. Participants complete two short fMRI runs (up to three minutes each) that alternate between emotion blocks and control blocks. In the emotion blocks, they match faces displaying fearful or angry expressions; in the control blocks, they match simple geometric shapes. Each block lasts approximately 18 seconds and includes several trials, isolating neural responses to emotional faces from general visual or matching processes. This task is widely used to study emotion reactivity, regulation, and individual differences across the large HCP sample \citep{barch2013function}. 

For our empirical analysis, we select a homogeneous sample of 12 female subjects ($K=12$, ages 22–30) from the Human Connectome Project (HCP) 3T protocol who completed all behavioral and imaging scans released in Quarter 1 (Q1). Cortical time series were extracted using the 400-parcel Schaefer2018 local–global atlas \citep{schaefer2018local}, mapped onto Yeo’s 17-network functional organization \citep{yeo2011organization}, yielding $d=17$ distinct cortical networks (e.g., Default Mode, Salience, Dorsal Attention, Somatomotor, and Visual networks). After removing transient artifacts at the scan boundaries, the final dataset consists of an average sample length of $T_k=165$ time points across all subjects.

\subsection{Application Results}

Figure \ref{fig:numer_data_appl} represents the results. The first row displays the estimated common paths across the four evaluated approaches: multi-VAR, multi-VAR(A), our proposed framework (Estimation), and our framework followed by hypothesis-based filtering (Tests). As shown in the third column, the proposed method identifies a sparser, more localized set of network paths compared to both the multi-VAR and multi-VAR (A) benchmarks. This identified network is further refined and focused by the subsequent hypothesis-testing stage (fourth column). The common paths recovered by the proposed framework form a structured subset of the multi-VAR (A), maintaining key directional patterns while eliminating diffuse background noise.

The second row summarizes the prevalence of unique nonzero paths across the 12 subjects, color-coded by their recurrence frequency (counts) across individuals. The multi-VAR benchmark produces a uniformly flat, saturated profile, indicating a high number of unique paths identified across almost all subjects (visible as a solid green distribution). While the multi-VAR(A) approach introduces some sparsity, it still yields a high density of idiosyncratic paths distributed across the network. In contrast, our proposed estimation and hypothesis-testing frameworks drastically reduce this background density. Under these methods, the unique paths are overwhelmingly sparse (represented primarily by the white and light orange regions), demonstrating that the framework successfully isolates subject-unique paths.

The third row reports the total density and distribution of individual nonzero paths across all participants. Mirroring the trends observed in the unique paths, the multi-VAR model yields completely saturated individual networks, resulting in an overly dense matrix that lacks structural interpretability. The multi-VAR(A) introduces structural heterogeneity but maintains a high baseline of nonzero paths. Conversely, our proposed estimation method and hypothesis-filtered framework significantly reduce network density under limited sample lengths, recovering interpretable individual network architectures by providing fewer paths.

Several highly persistent individual paths identified by our framework are consistent with established large-scale network interactions in emotion processing. For example, prominent pathways from Default Mode B and C to the Ventral Attention networks are consistent with the functional interaction between internally oriented processing and attentional systems involved in detecting behaviorally salient information \citep{sridharan2008critical}. The framework also identifies persistent paths from the Limbic subdivisions to Frontoparietal A, corresponding to the involvement of frontoparietal control systems in regulating affective processing through interactions with limbic systems \citep{ray2012anatomical}. Finally, the reciprocal paths between Limbic A and DefaultMode C are consistent with evidence of close functional integration between anterior temporal and orbitofrontal limbic regions and the default-mode network \citep{catani2013revised}.

\begin{figure}[th]
\centering
\includegraphics[width=1\textwidth,height=0.7\textheight]{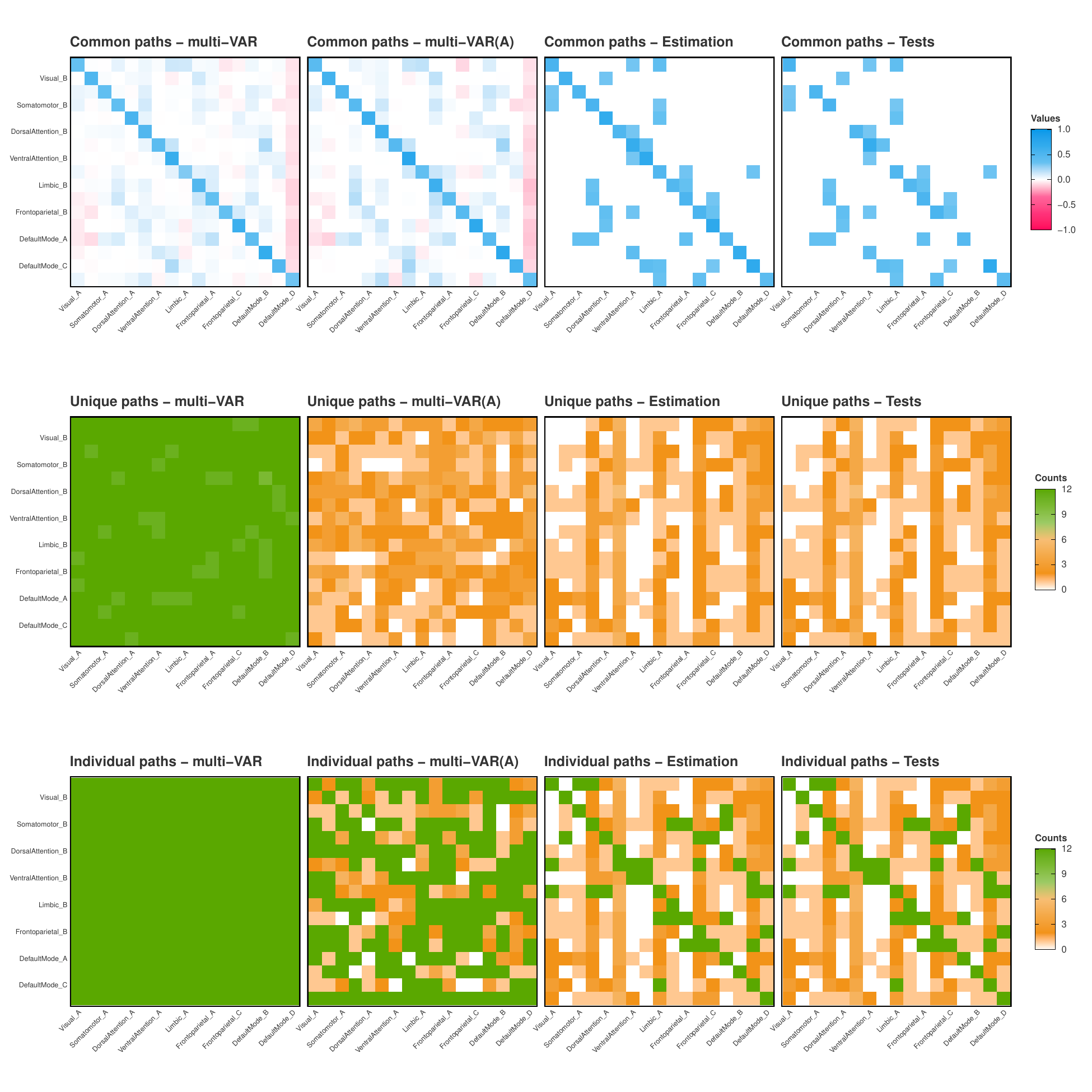}
\caption{Path identification from tfMRI on emotion processing. Each row presents results for common paths (top), unique paths (middle), and individual paths (bottom). The common paths are shown with their estimated values, while the unique and individual paths are summarized by the counts of nonzero entries across subjects. The first two columns report results from the stratified Lasso and its adaptive analogue (multi-VAR and multi-VAR(A), respectively), the third column shows results from the estimators of the proposed framework, and the last column presents the estimators after filtering through the hypothesis testing procedure. For each node in the common, unique, and individual path matrices, only entries for which the null hypotheses of significance, homogeneity, or nullity are rejected are colored.}
\label{fig:numer_data_appl}
\end{figure}

\section{Discussion}
\label{se:discuss}

Adopting the proposed restriction for the common path resolves the statistical inefficiency by allowing convergence rates to be tailored individually to each component's sparsity level and corresponding sample size. Across the simulation studies, the proposed multi-VARSE algorithm consistently achieves lower estimation error and superior specificity relative to the benchmark methods. The hypothesis testing framework controls the FDR reliably and maintains high power across all three test types, in close agreement with the asymptotic theory. The data application to the HCP emotion processing task fMRI reveals both shared and subject-unique cortical network paths.

Our framework can be extended in several promising directions. First, while it naturally handles sub-Gaussian innovations, adapting it to heavy-tailed distributions, such as sub-exponential errors, remains an open challenge. Second, while \cite{crawford2026penalized} recently explored capturing partially shared network paths using subject clustering in multi-VAR models, the statistical consistency of these models within a communication-efficient data framework remains unresolved. Finally, extending this approach to higher-lag VAR models presents a unique methodological hurdle. Adapting structured penalties in time series modeling to address this decay issue while ensuring the model can be properly debiased remains an important avenue for future research.

\section*{Acknowledgements}

Vladas Pipiras’s research was partially supported by the grants NSF DMS 1712966, DMS 2113662, and DMS 2134107.

\section*{Data Availability Statement}

HCP data used in the application in Section \ref{se:data} are publicly accessible. The dataset can be downloaded at \href{https://humanconnectome.org/}{https://humanconnectome.org/}
. The R code used in the simulations in Section \ref{se:illustrative} and in the data analysis in Section \ref{se:data} is available on GitHub at \href{https://github.com/yk748/multiVARSE}{https://github.com/yk748/multiVARSE}. 

{\small
\bibliographystyle{apalike}
\bibliography{multiVARSE}
}


\clearpage
\appendix

\setcounter{section}{0}
\setcounter{page}{1}
\setcounter{figure}{0}
\setcounter{equation}{0}
\renewcommand{\thesection}{S.\arabic{section}}
\renewcommand{\thepage}{\arabic{page}}
\renewcommand{\thetable}{S.\arabic{table}}
\renewcommand{\thefigure}{S.\arabic{figure}}
\renewcommand{\theequation}{S.\arabic{equation}}

\begin{center}
\textbf{\LARGE {Supplementary Material for \\ ``Joint modeling and inference of multiple-subject high-dimensional sparse vector autoregressive models}''}
\end{center}    

\begin{center}
Younghoon Kim$^{*1}$, Zachary F. Fisher$^2$, and Vladas Pipiras$^{3}$
\end{center}

\begin{center}
$^1$Department of Mathematics, University of Alabama \\
$^2$Department of Psychology and Neuroscience, University of North Carolina at Chapel Hill
$^3$Department of Statistics and Operations Research, University of North Carolina at Chapel Hill
\end{center}

\def\thefootnote{$*$}\footnotetext{Corresponding author. Email: ykim124@ua.edu}


Section \ref{se:theory_results} outlines the theoretical results and assumptions. The proofs for these results are provided in Section \ref{se:prop} and auxiliary lemmas are presented in Section \ref{se:lemmas}. Finally, Section \ref{se:additional} details the computation of the performance metrics used in Section \ref{se:illustrative} and provides additional simulation results.


\section{Theoretical Results}\label{se:theory_results}

In this section, we establish the theoretical properties of the estimation and inference frameworks introduced in Sections \ref{se:estimation} and \ref{se:tests}.

\subsection{Theory on Estimation}
\label{sse:theory_estimation}

In this section, we establish the convergence rates of the proposed estimators. We set $p=1$ and suppress the lag index in $\Phi_{\ell}^{(k)}$, writing simply $\Phi^{(k)}$. Note that similar assumptions for VAR$(p)$ models, $p>1$, can be defined analogously.

To present the result, we define three standard conditions commonly applied in high-dimensional time series (HDTS) modeling, as discussed in \citet{basu2024high_supp}.
\begin{itemize}
    \item[(a)] Stability regarding VAR transition matrix: for $\Phi^{(k)}(z)=I-\Phi_1^{(k)}z$ and $(\Phi^{(k)})^{-1}(z)=(I-\Phi^{(k)}(z))^{-1}$, $z\in\mathbb{C}$, consider 
    \begin{displaymath}
        |\!|\!|\Phi^{(k)}|\!|\!|:=\max_{|z|\leq1}\|\Phi^{(k)}(z)\|, \quad 
        |\!|\!|(\Phi^{(k)})^{-1}|\!|\!|:=\max_{|z|\leq1}\|(\Phi^{(k)})^{-1}(z)\|.
    \end{displaymath}
    The stability assumption implies the condition number $\kappa(\Phi^{(k)}):=|\!|\!|\Phi^{(k)}|\!|\!||\!|\!|(\Phi^{(k)})^{-1}|\!|\!|<\infty$.
    \item[(b)] Error covariance matrix of VAR model: for $\sigma_{k,\max}^{2}=\max_{j}\sigma_{k,j}^2$ and $\sigma_{k,\min}^{2}=\min_{j}\sigma_{k,j}^2>0$, the condition number is $\kappa(\Sigma_{\epsilon}^{(k)}):=\sigma_{k,\max}^2/\sigma_{k,\min}^2<\infty$.
    \item[(c)] Sparsity level: For $k\in[K]$, $s_{\Phi}^{(k)}=\|\Phi^{(k)}\|_0$ so that $s_{\Phi}^{\max}=\max_{k\in[K]} s_{\Phi}^{(k)}$. Also, for $\Theta^{(k)}=(\Sigma_X^{(k)})^{-1}$, define 
    \begin{equation}\label{e:sparse}
        (s_{\Theta}^{(k)})_j = \|\Theta_{j:}^{(k)}\|_0, \quad s_{\Theta}^{(k)} = \max_{j\in[d]} (s_{\Theta}^{(k)})_j.
    \end{equation}
\end{itemize}
Note that condition (c) is analogous to the assumption of the sparsity levels in \cite{maity2022meta_supp}, which inherits the theoretical framework developed by \cite{van2014asymptotically_supp}. In contrast, recent literature in the HDTS (e.g., \citet{basu2024high_supp}) adapts a weak sparsity condition on $\Theta^{(k)}$, originally introduced by \citet{javanmard2014confidence_supp}, and defined as:
\begin{displaymath}
    s_{\Theta}^{(k)} = \sum_{1 \leq i, j \leq d} |\Theta_{i,j}^{(k)}|.
\end{displaymath}
Although this formulation diverges from the recently established convention for HDTS, we demonstrate that a comparable theoretical convergence rate can still be achieved under the alternative condition, thereby enabling the foundational results of \cite{maity2022meta_supp} to be applied in the time-series context.

\begin{assumption}\label{ass:ts}
For each $k\in[K]$, we consider an asymptotic regime where $d,N_k\to\infty$,
\begin{equation}
    \kappa^2(\Sigma_{\epsilon}^{(k)})\kappa^4(\Phi^{(k)})|\!|\!|\Phi^{(k)}|\!|\!|^2\max\{s_{\Theta}^{(k)},s_{\Phi}^{(k)}\}\frac{\log d}{\sqrt{N_k}} \to 0. 
\end{equation}
\end{assumption}
\noindent As explained in \cite{basu2024high_supp}, this allows $d$ to grow with $N_k$ as long as $\max\{s_{\Theta}^{(k)},s_{\Phi}^{(k)}\}$ grow as $o(\sqrt{N_k})$. If the eigenvalues of $\Sigma_{\epsilon}^{(k)}$ and the modulus of eigenvalues of $\Phi^{(k)}(z)$ are both bounded away from zero and infinity, the terms involving $\Sigma_{\epsilon}^{(k)}$ and $\Phi^{(k)}$ will not appear in the convergence analysis and match the known error bounds in the high-dimensional regression with i.i.d. data.

Next, we introduce a set of other conditions.
\begin{assumption}\label{ass:aggregation}
The following conditions are assumed to be held.
\begin{enumerate}
    \item[(a)] Let $\mathcal{I}_{(i,j)}\subseteq[K]$ be a designated population inlier set, 
    and define
    \begin{displaymath}
        \mu_{(i,j)} 
        = \frac{1}{|\mathcal{I}_{(i,j)}|}\sum_{k\in \mathcal{I}_{(i,j)}}(\beta_{i}^{(k)})_j.
    \end{displaymath}
    We assume $\frac{|\mathcal{I}_{(i,j)}|}{K} \geq \frac{4}{7}$.
    \item[(b)] Let $\delta_{\textrm{in}},\delta_{\textrm{out}}$, $0\leq\delta_{\textrm{in}}<\delta_{\textrm{out}}$, 
    $\delta_{\textrm{out}} \geq 5\delta_{\textrm{in}}$, be sufficiently small nonnegative real numbers 
    that distinguish inliers and outliers centered at $\mu_{(i,j)}$. We assume that
    \begin{displaymath}
        (\beta_{i}^{(k)})_j\in[\mu_{(i,j)}-\delta_{\textrm{in}},\mu_{(i,j)}+\delta_{\textrm{in}}]
    \end{displaymath}
    for all $k\in \mathcal{I}_{(i,j)}$, and none of $(\beta_{i}^{(k)})_j$'s for $k\notin\mathcal{I}_{(i,j)}$ 
    lie in the intervals $[\mu_{(i,j)}-\delta_{\textrm{out}},\mu_{(i,j)}-\delta_{\textrm{in}})$ 
    or $(\mu_{(i,j)}+\delta_{\textrm{in}},\mu_{(i,j)}+\delta_{\textrm{out}}]$.

    \item[(c)] Define the minimum separation between inlier and outlier values as
    \begin{displaymath}
        \delta_{\min} 
        = \min_{k_1 \in \mathcal{I}_{(i,j)}, k_2\notin \mathcal{I}_{(i,j)}} |(\beta_{i}^{(k_1)})_j - (\beta_{i}^{(k_2)})_j| 
    \end{displaymath}
    if $\mathcal{I}_{(i,j)}^c \neq \emptyset$ and $\delta_{\min} =\infty$ if $\mathcal{I}_{(i,j)}^c = \emptyset$. Let $\delta_{\textrm{gap}}=\delta_{\textrm{out}}-\delta_{\textrm{in}}$. We assume that $\eta_{i,j}$ satisfies $2\delta_{\textrm{in}} < \eta_{i,j} < \frac{\delta_{\min}}{2}$.
\end{enumerate}
\end{assumption}
\noindent Note that (a)--(c) in Assumption \ref{ass:aggregation} follow Assumption 4 of \citet{maity2022meta} for their Result 5 therein, which validates the existence of the unique minimizer through an identification step. Condition (a) defines $\mathcal{I}_{(i,j)}$ as a designated population quantity, independent of the minimizer $(\alpha_i^{(0)})_j$, avoiding circularity; the case $\mathcal{I}_{(i,j)}^c = \emptyset$ corresponds to the homogeneous setting where all subjects share the same path value, for which $\delta_{\min} = \infty$ renders the upper bound in condition (c) vacuous. 
To see why $|\mathcal{I}_{(i,j)}|/K \geq 4/7$ is required, note that a sufficient condition for $\mu_{(i,j)}$ to be the unique global minimizer of $L_{i,j}$ is
\begin{equation}\label{e:dominance_simple}
    \frac{|\mathcal{I}_{(i,j)}|}{K} > \frac{1}{2 - (\delta_{\textrm{in}}/\eta_{i,j})^2}.
\end{equation}
Under conditions (b) and (c), $\eta_{i,j} > 2\delta_{\textrm{in}}$ holds, which makes $(\delta_{\textrm{in}}/\eta_{i,j})^2 < 1/4$ and reduces the right-hand side of \eqref{e:dominance_simple} to below $4/7$, so condition (a) is sufficient.

Now, we show the convergence rates of the estimated common and unique paths.
\begin{proposition}\label{prop:estim}
Suppose that Assumptions \ref{ass:ts} and \ref{ass:aggregation} hold. For $e\in(0,3/7]$, set $\delta_0=\max_{k\in[K]}\kappa(\Sigma_{\epsilon}^{(k)})\sqrt{\frac{\log d^2}{(1-e)KN_{\min}}}$, and $\delta_k = \kappa(\Sigma_{\epsilon}^{(k)})\sqrt{\frac{\log d^2}{N_k}}$ in \eqref{e:hard_thres}. For sufficiently large $N_k$, $k\in[K]$, we have the following:
\begin{enumerate}
    \item[(a)] $\|\hat{\alpha}^{(0)} - \alpha^{(0)}\|_{2} = \mathcal{O}_{\mathbb{P}}\left(\sqrt{\frac{\|\alpha^{(0)}\|_0 \log d^2}{K N_{\min}}}\right)$,
    \item[(b)] $\|\hat{\alpha}^{(k)} - \alpha^{(k)}\|_{2} = \mathcal{O}_{\mathbb{P}}\left(\sqrt{\frac{\|\alpha^{(k)}\|_0 \log d^2}{N_{k}}}\right)$, $k\in[K]$.
\end{enumerate} 
\end{proposition}
\noindent The convergence rates of both the common and unique components are now adjusted separately, with the common path benefiting from averaging over $K$ subjects through the factor $KN_{\min}$ in part (a). The sparsity levels $\|\alpha^{(0)}\|_0$ and $\|\alpha^{(k)}\|_0$ are irrelevant to $s_{\Phi}^{(k)}$ and $s_{\Theta}^{(k)}$ in Assumption \ref{ass:ts}; this is a consequence of hard thresholding in the identification step via Lemma \ref{lem:lee}. The role of $e$ is analogous to $\alpha_{\max}$ in \citet{maity2022meta_supp}, where $|\mathcal{I}_{(i,j)}|=(1-e)K$ captures the outlier proportion. The constraint $e\in(0,3/7]$ follows directly from 
condition (a) of Assumption \ref{ass:aggregation}. This factor does not substantially affect the asymptotic rate.

\subsection{Theory on Inference}
\label{sse:theory_test}

In this section, we present theoretical results on hypothesis testing. Specifically, Proposition \ref{prop:test_single} establishes the asymptotic distributions of the test statistics for the single-test procedure. Proposition \ref{prop:test_multiple} provides guarantees on the control of FDR and the power of the multiple hypothesis testing procedures. Along with the theory on estimation in Section \ref{sse:theory_estimation}, we assume $p=1$ without loss of generality.

First, we introduce the preliminary result, which is used as a building block for the proposed hypothesis tests in Section \ref{sse:test_single}. Let
$\tilde{\beta}_{(i,j)}=((\tilde{\beta}_{i}^{(1)})_j,\ldots,(\tilde{\beta}_{i}^{(k)})_j)'\in\mathbb{R}^K$ and $\beta_{(i,j)}=((\beta_{i}^{(1)})_j,\ldots,(\beta_{i}^{(k)})_j)'\in\mathbb{R}^K$ be $K$-dimensional estimators of $j^{\textrm{th}}$ entries in $i^{\textrm{th}}$ variable across $K$ subjects and its population analog. It is known that \citep{adamek2023lasso_supp,basu2024high_supp} for $i\in[d],j\in[dp]$, and $k\in[K]$, the nullity of the coefficient in the transition matrix of each subject can be tested as follows: the hypotheses are 
\begin{displaymath}
    H_{0,(i,j)}^{(k)}: (\beta_{i}^{(k)})_j=0 \quad vs \quad 
    H_{1,(i,j)}^{(k)}: (\beta_{i}^{(k)})_j\neq 0.
\end{displaymath}
Then, Proposition 2.3 in \cite{basu2024high_supp} implies that
\begin{equation}\label{e:clt_betak}
    Z_{(i,j)}^{(k)}((\beta_{i}^{(k)})_j) = \frac{\sqrt{N_k}((\tilde{\beta}_{i}^{(k)})_j - (\beta_{i}^{(k)})_j)}{\sqrt{\hat{V}_{i,j}^{(k)}}}\stackrel{d}{\to}\mathcal{N}(0,1).
\end{equation}
Under $H_{0,(i,j)}^{(k)}$, we reject the null hypothesis if $|Z_{(i,j)}^{(k)}(0)| \geq z_{1-\alpha/2}$ for a significance level $\alpha$.

Now we state the theoretical results on the proposed hypothesis tests in Section \ref{sse:test_single}.

\begin{proposition}\label{prop:test_single}
Suppose that Assumptions \ref{ass:ts} and \ref{ass:aggregation} hold. For the three specific hypothesis tests in Section \ref{sse:test_single},
\begin{enumerate}
    \item[(a)] Under the null hypothesis $H_{0,(i,j)}^{(0)}$ in \eqref{e:h0_common}, the test statistic $Z_{(i,j)}$ in \eqref{e:clt_alpha0} asymptotically follows the standard normal distribution.
    \item[(b)] Under the null hypothesis $H_{0,(i,j)}^{[K]}$ in \eqref{e:h0_homo}, the test statistic $Z_{(i,j)}$ in \eqref{e:z_homo} asymptotically follows the standard normal distribution.
    \item[(c)] Under the null hypothesis $H_{0,(i,j)}^{[K]_{\circ}}$ in \eqref{e:h0_null}, the test statistic $Z_{(i,j)}$ in \eqref{e:z_null} asymptotically follows the standard normal distribution.
\end{enumerate}
\end{proposition}

Next, we state the results on the FDR controls on the hypothesis tests in Section \ref{sse:test_multiple}. For some threshold $\zeta$, the power (true positive proportion) and FDR, defined via the false discovery proportion (FDP) as $\textrm{FDR}(\zeta)=\mathbb{E}(\textrm{FDP}(\zeta))$ of the test are obtained by
\begin{align}
    \textrm{FDP}(\zeta) 
    &= \frac{\sum_{(i,j)\in \mathcal{H}_{0}} 1_{\{|Z_{(i,j)}| > \zeta\}}}{\max\{1, \sum_{(i,j)\in\mathcal{S}} 1_{\{|Z_{(i,j)}| > \zeta\}}\}}, \label{e:fdp}\\
    \textrm{Power}(\zeta) 
    &= \sum_{(i,j)\in \mathcal{H}_{1}}\frac{1_{\{|Z_{(i,j)}| > \zeta\}}}{|\mathcal{H}_{1}|}, \label{e:power}
\end{align}
where $\mathcal{H}_{0}$ and $\mathcal{H}_{1}$ are the sets of true null and true alternative hypotheses, respectively. For each test, the sets are defined as
\begin{itemize}
    \item Test of significance:
    \begin{displaymath}
        \mathcal{H}_{0}^{(0)} 
        = \{(i,j) \in \mathcal{S}: (\alpha_{i}^{(0)})_j =0 \}, \quad
        \mathcal{H}_{1}^{(0)} 
        = \{(i,j) \in \mathcal{S}: (\alpha_{i}^{(0)})_j \neq 0 \}.
    \end{displaymath}
    \item Test of homogeneity:
    \begin{displaymath}
        \mathcal{H}_{0}^{[K]} 
        = \{(i,j) \in \mathcal{S}: (\beta_{i}^{(1)})_j=\ldots=(\beta_{i}^{(k)})_j \}, \quad
        \mathcal{H}_{1}^{[K]} 
        = \{(i,j) \in \mathcal{S}: \exists k,\ell\in[K] \ \textrm{s.t.} \ (\beta_{i}^{(k)})_j\neq(\beta_{i}^{(\ell)})_j \}.
    \end{displaymath}
    \item Test of nullity: 
    \begin{displaymath}
        \mathcal{H}_{0}^{[K]_{\circ}}  
        = \{(i,j) \in \mathcal{S}: (\beta_{i}^{(1)})_j=\ldots=(\beta_{i}^{(k)})_j=0 \}, \quad
        \mathcal{H}_{1}^{[K]_{\circ}} 
        = \{(i,j) \in \mathcal{S}: \exists k\in[K] \ \textrm{s.t.} \ (\beta_{i}^{(k)})_j\neq0 \}.
    \end{displaymath}
\end{itemize}
For example, if we test the significance, then we replace $\mathcal{H}_0$ and $\mathcal{H}_1$ with $\mathcal{H}_0^{(0)}$ and $\mathcal{H}_1^{(0)}$, respectively.

The following assumption is needed for simultaneous FDR control and guaranteed power.
\begin{assumption}\label{ass:FDR_power}
    Consider the sets of true null and true alternative hypotheses $\mathcal{H}_0$ and $\mathcal{H}_1$. Then
    \begin{enumerate}
        \item[(a)] With $c_1>2$, the following holds: With $W_{(i,j)}^{(0)}=\sqrt{\frac{1}{|\mathcal{J}_{(i,j)}|^2}\sum_{k\in\mathcal{J}_{(i,j)}}V_{i,j}^{(k)}}$,
        \begin{displaymath}
            \min_{(i,j)\in\mathcal{S}} \left|\frac{\frac{1}{|\mathcal{J}_{(i,j)}|}\sum_{k\in\mathcal{J}_{(i,j)}}  (\alpha_{i}^{(0)})_j}{\sqrt{W_{(i,j)}^{(0)}/N_{(i,j)}}}\right|
            \geq c_1\sqrt{2\log|\mathcal{S}|}.
        \end{displaymath}
        for the test of significance \eqref{e:h0_common_multi}, 
        \begin{displaymath}
            \min_{(i,j)\in\mathcal{S}} \max_{k \neq \ell} \left|\frac{(\beta_{i}^{(k)})_j - (\beta_{i}^{(\ell)})_j}{\sqrt{\frac{1}{N_k}V_{i,j}^{(k)} + \frac{1}{N_{\ell}}V_{i,j}^{(\ell)} }} \right|
            \geq c_1\sqrt{2\log|\mathcal{S}|}.
        \end{displaymath}
        for the test of homogeneity \eqref{e:h0_homo_multi}, and
        \begin{displaymath}
            \min_{(i,j)\in\mathcal{S}} \max_{k\in[K]} \left| \frac{\sqrt{N_k}(\beta_{i}^{(k)})_j}{\sqrt{V_{i,j}^{(k)}}}\right|
            \geq c_1\sqrt{2\log|\mathcal{S}|}.
        \end{displaymath}
        for the test of nullity \eqref{e:h0_null_multi}.
        \item[(b)] $c_2\sqrt{\log|\mathcal{S}|}=o(1)$ for some sequence $c_2\to0$.
        \item[(c)] $|\mathcal{H}_1| \geq c_3\sqrt{\log\log|\mathcal{S}|}$ for some $c_{3}>0$.
    \end{enumerate}
\end{assumption}
\noindent Note that (a) in Assumption \ref{ass:FDR_power} is a required minimum signal strength requirement that every true alternative is sufficiently separated from the null. (b) in Assumption \ref{ass:FDR_power} ensures that the remainder in the Gaussian approximation is negligible relative to the signal at the threshold scale. (c) in Assumption \ref{ass:FDR_power} requires the number of true alternatives to grow, ensuring sufficiently many rejections to drive into the interior of $(0,\sqrt{2\log|\mathcal{S}|}]$.

\begin{proposition}\label{prop:test_multiple}
Assume that the conditions in Assumption \ref{ass:FDR_power} and the same conditions as Proposition \ref{prop:test_single} hold. For all three types of tests, we have
\begin{equation}\label{e:fdp_power}
    \frac{\textrm{FDP}(\hat{\zeta})}{\alpha|\mathcal{H}_0|/|\mathcal{S}|} \stackrel{p}{\to} 1 \quad \textrm{and} \quad
    \textrm{Power}(\hat{\zeta}) \stackrel{p}{\to} 1 \ \textrm{as }|\mathcal{S}|\to\infty.
\end{equation}
\end{proposition}


\section{Proofs of Theoretical Results}\label{se:prop}

In this section, we provide proofs of the main results described in Section \ref{se:theory_results}.


\subsection{Proofs of Proposition \ref{prop:estim}}

From Lemma \ref{lem:lee}, it remains to show that values of $\delta_k$, $k\in[\bar{K}]$, satisfy the condition in Proposition \ref{prop:estim}, which is directly derived from Lemma \ref{lem:alpha_bdd}. Finally, taking $s_{\theta}=\|\alpha^{(0)}\|_0$ for (a) and $s_{\theta}=\|\alpha^{(k)}\|_0$, $k\in[K]$ for (b) yields the desired result.  \qed


\subsection{Proofs of Proposition \ref{prop:test_single}}

First, consider (a). It suffices to show that \eqref{e:clt_alpha0} holds. Define the indicator that the element of $k^{\textrm{th}}$ subject contributes to the common path,
\begin{displaymath}
    1_{i,j}^{(k)} := 1_{\{|(\beta_{i}^{(k)})_j - (\alpha_{i}^{(0)})_j| \leq \eta_{i,j}\}}.
\end{displaymath}
From (a)--(c) in Assumption \ref{ass:aggregation}, the minimizer in \eqref{e:redescending_estim} is unique, and none of $|(\beta_{i}^{(k)})_j - (\alpha_{i}^{(0)})_j|=\eta_{i,j}$ holds for all $k$. Then the derivative of the objective function in \eqref{e:redescending_estim} is defined by
\begin{displaymath}
    L_{i,j}'(x) = -2\sum_{k=1}^K ((\tilde{\beta}_{i}^{(k)})_j-x)1_{\{|(\tilde{\beta}_{i}^{(k)})_j-x|\leq \eta_{i,j}\}},
\end{displaymath}
and it satisfies $L_{i,j}'((\tilde{\alpha}_{i}^{(0)})_j)=0$. So, by Taylor expansion around $(\alpha_{i}^{(0)})_j$,
\begin{displaymath}
    0 = L_{i,j}'(\alpha_{i}^{(0)}) + L_{i,j}''((\alpha_{i}^{(0)})_j)((\tilde{\alpha}_{i}^{(0)})_j- (\alpha_{i}^{(0)})_j).
\end{displaymath}
Note that $1_{\{|(\tilde{\beta}_{i}^{(k)})_j -(\alpha_{i}^{(0)})_j|\leq \eta_{i,j}\}} \stackrel{p}{\to} 1_{i,j}^{(k)}$. With $N_{(i,j)}=(\sum_{k=1}^K 1_{i,j}^{(k)} N_k)/(\sum_{k=1}^K 1_{i,j}^{(k)})$,
\begin{align*}
    L_{i,j}'((\alpha_{i}^{(0)})_j) 
    &= -2\sum_{k=1}^K 1_{i,j}^{(k)}((\tilde{\beta}_{i}^{(k)})_j - (\alpha_{i}^{(0)})_j) + \mathcal{O}_{\mathbb{P}}( N_{i,j}^{-1/2} ), \\
    L_{i,j}''((\alpha_{i}^{(0)})_j) &= 2\sum_{k=1}^K 1_{i,j}^{(k)} + \mathcal{O}_{\mathbb{P}}( 1 ).
\end{align*}
Therefore, 
\begin{displaymath}
    (\tilde{\alpha}_{i}^{(0)})_j- (\alpha_{i}^{(0)})_j 
    = \frac{\sum_{k=1}^K 1_{i,j}^{(k)}((\tilde{\beta}_{i}^{(k)})_j - (\alpha_{i}^{(0)})_j )}{\sum_{k=1}^K 1_{i,j}^{(k)}} + \mathcal{O}_p( N_{(i,j)}^{-1/2} )
    = \frac{\sum_{k=1}^K 1_{i,j}^{(k)}((\tilde{\beta}_{i}^{(k)})_j - (\beta_{i}^{(k)})_j )}{\sum_{k=1}^K 1_{i,j}^{(k)}} + \mathcal{O}_p( N_{(i,j)}^{-1/2} ),
\end{displaymath}
since $\sum_{k=1}^K 1_{i,j}^{(k)}((\beta_{i}^{(k)})_j -(\alpha_{i}^{(0)})_j)=0$. Therefore, by Slutsky’s theorem,
\begin{displaymath}
    \sqrt{N_{(i,j)}}((\tilde{\alpha}_{i}^{(0)})_j- (\alpha_{i}^{(0)})_j) \stackrel{d}{\to} \mathcal{N}(0,W_{(i,j)}^{(0)}),
\end{displaymath}
where 
\begin{equation}\label{e:var_together}
    \hat{W}_{(i,j)}^{(0)}
    = \sum_{k=1}^K 1_{ij}^{(k)}\hat{V}_{i,j}^{(k)}/(\sum_{k=1}^K 1_{i,j}^{(k)})^2 
    \stackrel{p}{\to} W_{(i,j)}^{(0)}
\end{equation}
by combining the continuous mapping theorem and Lemma \ref{lem:scaled_var}. Hence, the test statistic $Z_{(i,j)}$ in \eqref{e:clt_alpha0} also follows the standard normal distribution under the null hypothesis $H_{0,(i,j)}^{(0)}$.

Next, consider (b). Under the assumption \eqref{ass:ts}, the test statistic for the individual debiased estimator $(\tilde{\beta}_{i}^{(k)})_j$ follows \eqref{e:clt_betak}. Stacking all $\binom{K}{2}$ pairwise statistics into the vector $\mathbf{Z}_{(i,j)} = \left(Z_{(i,j)}^{(k,\ell)}\right)_{1\leq k < \ell \leq K}$, Note that the asymptotic covariance between two pairwise statistics as $\Sigma_{(i,j)}$. $(k_1\ell_1,k_2\ell_2)^{\textrm{th}}$ entry of the estimator $\hat{\Sigma}_{(i,j)}$ is obtained by
\begin{displaymath}
    \frac{\frac{1}{N_{k_1}}\hat{V}_{i,j}^{(k_1)}1_{\{k_1=k_2\}} 
    + \frac{1}{N_{\ell_1}}\hat{V}_{i,j}^{(\ell_1)}1_{\{\ell_1=\ell_2\}}
    - \left(\frac{1}{N_{k_1}}\hat{V}_{i,j}^{(k_1)}1_{\{k_1=\ell_2\}}
    + \frac{1}{N_{\ell_1}}\hat{V}_{i,j}^{(\ell_1)}1_{\{\ell_1=k_2\}} \right)}
    {\sqrt{\frac{\hat{V}_{i,j}^{(k_1)}}{N_{k_1}}+\frac{\hat{V}_{i,j}^{(\ell_1)}}{N_{\ell_1}}}\sqrt{\frac{\hat{V}_{i,j}^{(k_2)}}{N_{k_2}}+\frac{\hat{V}_{i,j}^{(\ell_2)}}{N_{\ell_2}}}},
\end{displaymath}
$k_1,\ell_1\in[K]$, $k_1<\ell_1$, and $k_2,\ell_2\in[K]$, $k_2<\ell_2$. Note that the rank of $\Sigma_{(i,j)}$ satisfies $M = \mathrm{rank}(\Sigma_{(i,j)}) \leq K-1$ since all $K(K-1)/2$ pairwise differences $(\beta_{i}^{(k)})_j-(\beta_{i}^{(\ell)})_j$ lie in a $(K-1)$-dimensional subspace of $\mathbb{R}^K$. Applying Cram\'{e}r-Wold theorem and Lemma \ref{lem:clt_beta_pair} together to $\mathbf{Z}_{(i,j)}$ yields 
\begin{displaymath}
    \mathbf{Z}_{(i,j)} 
    \stackrel{d}{\to}\mathcal{N}(0,\Sigma_{(i,j)}),
\end{displaymath}
where $\Sigma_{(i,j)}$ is the $\binom{K}{2}\times\binom{K}{2}$-dimensional asymptotic covariance matrix. The off-diagonal entries are nonzero since pairwise differences share the subjects, so are their correlations. Since $\chi_{(i,j)}^2 = \|\mathbf{Z}_{(i,j)}\|_2^2$ is a continuous function of $\mathbf{Z}_{(i,j)}$, by using continuous mapping theorem, 
\begin{displaymath}
    \chi_{(i,j)}^2 \stackrel{d}{\to} \|\mathbf{W}\|_2^2, \quad \mathbf{W}\sim\mathcal{N}(0,\Sigma_{(i,j)}).
\end{displaymath}
Take the eigendecomposition of $\Sigma_{(i,j)}=P\Lambda P'$, where $\Lambda=\textrm{diag}(\nu_1,\ldots,\nu_M,0,\ldots,0)\in\mathbb{R}^{\binom{K}{2}}$, where $M\leq K-1$ since all $\binom{K}{2}$ pairwise difference $(\beta_{i}^{(k)})_j-(\beta_{i}^{(\ell)})_j$ lies in the $K-1$ dimensional subspace of $\mathbb{R}^k$. Let $\mathbf{Z} = (Z_1,\ldots,Z_{\binom{K}{2}})'\in\mathbb{R}^{\binom{K}{2}}$ such that $\mathbf{Z}\sim\mathcal{N}(0,I_M)$. Then we have
\begin{displaymath}
    \mathbf{W} 
    = \mathbf{Z}'\Lambda \mathbf{Z}
    = \sum_{m=1}^M \nu_m Z_{m}^2, \quad Z_1,\ldots,Z_{M}\stackrel{i.i.d.}{\sim}\mathcal{N}(0,1).
\end{displaymath}
Note that the weighted chi-square distribution is continuous on $[0,\infty)$ as long as $\nu_m>0$, and $\mathcal{F}_0$ is continuous and strictly increasing on its support. Therefore, under the null hypothesis $H_{0,(i,j)}^{[K]}$,
\begin{displaymath}
    \mathcal{F}_0\left(\sum_{m=1}^M \nu_m Z_{m}^2\right) \stackrel{d}{:\equiv} U \sim \textrm{Unif}(0,1).
\end{displaymath}
Applying $\Phi^{-1}(\cdot)$, which is also a continuous and strictly increasing function, gives
\begin{displaymath}
    Z_{(i,j)} 
    = \Phi^{-1}\left(\mathcal{F}_0\left(\chi_{(i,j)}^2\right)\right)
    \stackrel{d}{\to} \Phi^{-1}(U) = \mathcal{N}(0,1).
\end{displaymath}
It remains to show that the estimated covariance $\hat{\Sigma}_{(i,j)}$ is convergent. From Lemma \eqref{lem:scaled_var}, $\hat{V}_{i,j}^{(k)} \stackrel{p}{\to} V_{i,j}^{(k)}$ for all $i,j$, and $k$. Since the $\hat{\Sigma}_{(i,j)}$ consists of $\hat{V}_{i,j}^{(k)}$ and $N_k$, by the continuous mapping theorem, $\hat{\Sigma}_{(i,j)} \stackrel{p}{\to}\Sigma_{(i,j)}$. Therefore, the eigenvalues from $\hat{V}_{i,j}^{(k)}$ converges to those from $\Sigma_{(i,j)}$, and the weighted chi-square distribution $\mathcal{F}_0(x)$ is well-defined.

Finally, consider (c). Since the statistics of the debiased estimator $(\tilde{\beta}_{i}^{(k)})_j$ forms $Z_{(i,j)}$ in \eqref{e:z_null} independently, from the result of \eqref{e:clt_betak}, the test statistic satisfies
\begin{displaymath}
    \chi_{(i,j)}^2 
    = \sum_{k=1}^K\left(Z_{(i,j)}^{(k)}(0)\right)^2  
    \stackrel{d}{\to} \sum_{k=1}^K Z_{k}^2, \quad Z_1,\ldots,Z_{K} \stackrel{i.i.d.}{\sim}\mathcal{N}(0,1)
\end{displaymath}
by the continuous mapping theorem. Note that the limiting distribution of $\sum_{k=1}^K Z_{k}^2$ is exactly the chi-square distribution with $K$ degrees of freedom, with the CDF $\mathcal{F}_K(x)$. For $X\sim\chi_K^2$,  $\mathcal{F}_K(X)\sim\mathrm{Unif}(0,1)$. Therefore, under the null hypothesis $H_{0,(i,j)}^{[K]}$,
\begin{displaymath}
    \mathcal{F}_K\left(\chi_{(i,j)}^2\right)
    \stackrel{d}{\to} \mathcal{F}_K\left(\sum_{k=1}^KZ_k^2\right) 
    \stackrel{d}{\equiv} U\sim\mathrm{Unif}(0,1).
\end{displaymath}
Applying $\Phi^{-1}(\cdot)$ again yields
\begin{displaymath}
    Z_{(i,j)} = \Phi^{-1}\left(\mathcal{F}_K\left(\chi_{(i,j)}^2\right)\right) 
    \stackrel{d}{\to} \Phi^{-1}(U) = \mathcal{N}(0,1),
\end{displaymath}
which completes the proof. \qed


\subsection{Proofs of Proposition \ref{prop:test_multiple}}

First, consider the test of significance so that take $\mathcal{H}_{0}=\mathcal{H}_{0}^{(0)}$ and $\mathcal{H}_{1}=\mathcal{H}_{1}^{(0)}$. Note that 
\begin{displaymath}
    Z_{(i,j)} 
    = \frac{\sqrt{N_{(i,j)}}(\tilde{\alpha}_{i}^{(0)})_j}{\sqrt{\hat{W}_{(i,j)}^{(0)}}} 
    = \frac{\sqrt{N_{(i,j)}}\left( \frac{\sum_{k}1_{i,j}^{(k)}((\tilde{\beta}_{i}^{(k)})_j - (\beta_{i}^{(k)})_j) + \sum_k1_{i,j}^{(k)}(\alpha_{i}^{(0)})_j }{\sum_{k}1_{i,j}^{(k)}} \right)}{\sqrt{\hat{W}_{(i,j)}^{(0)}}} + o_{\mathbb{P}}(1), 
\end{displaymath} 
where $\sqrt{N_k}((\tilde{\beta}_{i}^{(k)})_j - (\beta_{i}^{(k)})_j)$ is expanded as
\begin{displaymath}
    \sqrt{N_k}((\tilde{\beta}_{i}^{(k)})_j - (\beta_{i}^{(k)})_j) 
    = \underbrace{\frac{1}{\sqrt{N_k}}\Theta_{j:}\mathcal{X}^{(k)'}E_{i}^{(k)}}_{:=u_{i,j}^{(k)}}
    + \underbrace{\frac{1}{\sqrt{N_k}}(\hat{\Theta}_{j:}-\Theta_{j:})\mathcal{X}^{(k)'}E_{i}^{(k)} + \sqrt{N_k}\Delta_{i,j}^{(k)}}_{:=r_{i,j}^{(k)}}.
\end{displaymath}
Substituting and separating the term yields
\begin{displaymath}
    Z_{(i,j)} 
    \leq \underbrace{\frac{\frac{1}{|\mathcal{J}_{(i,j)}|} \sum_{k}1_{i,j}^{(k)}(u_{i,j}^{(k)} + r_{i,j}^{(k)})}{\sqrt{\hat{W}_{(i,j)}^{(0)}}}}_{:=\mathcal{Z}_{(i,j)}} + \sqrt{\frac{W_{(i,j)}^{(0)}}{\hat{W}_{(i,j)}^{(0)}}}\frac{\sqrt{N_{(i,j)}}\sum_k 1_{i,j}^{(k)}(\alpha_{i}^{(0)})_j }{|\mathcal{J}_{(i,j)}|\sqrt{W_{(i,j)}^{(0)}}}.
\end{displaymath}
Note that $W_{(i,j)}^{(0)}/\hat{W}_{(i,j)}^{(0)} \to 1$ by continuous mapping theorem. By Lemma \ref{lem:martingale_clt}, $u_{i,j}^{(k)} \stackrel{d}{\to} \mathcal{N}(0,V_{i,j}^{(k)})$ 
for each $k\in\mathcal{J}_{(i,j)}$. By using Slutsky's theorem, we have
\begin{displaymath}
    \frac{1}{|\mathcal{J}_{(i,j)}|}\sum_{k \in \mathcal{J}_{(i,j)}}u_{i,j}^{(k)} \to \mathcal{N}\left(0,W_{(i,j)}^{(0)}\right).
\end{displaymath}
Also, from \eqref{e:second_term} in the proof of Lemma \ref{lem:beta_tilde_bdd} and the result of Lemma \eqref{lem:delta_bdd}, we have
\begin{displaymath}
    |r_{i,j}^{(k)}| 
    = \mathcal{O}_{\mathbb{P}}\left(\kappa^2(\Sigma_{\epsilon}^{(k)})\kappa^4(\Phi^{(k)})|\!|\!|\Phi^{(k)}|\!|\!|^2 s_{\Phi}^{(k)}\frac{\log d}{\sqrt{N_k}}\right) = o_{\mathbb{P}}(1)
\end{displaymath}
where the last equality holds due to Assumption \ref{ass:ts}. Consequently, 
\begin{displaymath}
    \frac{1}{|\mathcal{J}_{(i,j)}|}\sum_{k \in \mathcal{J}_{(i,j)}}r_{i,j}^{(k)} 
    = o_{\mathbb{P}}(1).
\end{displaymath}
Finally, using $W_{(i,j)}^{(0)}/\hat{W}_{(i,j)}^{(0)} \to 1$, we have $\mathcal{Z}_{(i,j)}\to\mathcal{N}(0,1)$. Then, by using $\hat{\zeta}\leq \sqrt{2\log|\mathcal{S}|}$ and (a) and (b) in Assumption \ref{ass:FDR_power}, 
\begin{align*}
    \max_{(i,j)\in\mathcal{S}}\mathbb{P}(Z_{(i,j)} \leq \hat{\zeta}) 
    & \leq  \max_{(i,j)\in\mathcal{S}}\mathbb{P} \left( \mathcal{Z}_{(i,j)}\leq \hat{\zeta} -\min_{(i,j)\in\mathcal{S}} \frac{\sqrt{N_{(i,j)}}}{|\mathcal{J}_{(i,j)}|}\frac{\sum_{k}1_{i,j}^{(k)}(\alpha_{i}^{(0)})_j}{\sqrt{W_{(i,j)}^{(0)}}} + c_2\right) \\
    & \leq \max_{(i,j)\in\mathcal{S}}\Phi \left(-(c_1-1)\sqrt{2\log(|\mathcal{S}|)} +c_2 \right) \to 0
\end{align*}
as $|\mathcal{S}|\to\infty$. Hence, from the definition power \eqref{e:power}, we have
\begin{displaymath}
    \textrm{Power}(\hat{\zeta}) 
    = \sum_{(i,j)\in \mathcal{H}_{1}^{(0)}}\frac{1_{\{|Z_{(i,j)}| > \hat{\zeta}\}}}{|\mathcal{H}_{1}^{(0)}|} 
    \geq 1 - \max_{(i,j)\in\mathcal{S}}\mathbb{P}(Z_{(i,j)} \leq \hat{\zeta}) \stackrel{p}{\to}1.
\end{displaymath}
Also, by the maximum bound for sub-Gaussian random variables \cite[e.g., Chapter 2.5.2 in][]{vershynin2018high_supp},
\begin{displaymath}
    \mathbb{P}\left(\max_{(i,j)\in\mathcal{S}} |Z_{(i,j)}| > \sqrt{4\log d}\right) \to 0.
\end{displaymath}
Since $\mathbb{P}(|Z_{(i,j)}|>\sqrt{2\log|\mathcal{S}|})\to1$ for $(i,j)\in\mathcal{H}_1^{(0)}$, we have
\begin{align*}
    \sum_{(i,j)\in\mathcal{S}} 1_{ \{|Z_{(i,j)}|>\sqrt{2\log(|\mathcal{S}|)} \} } 
    &\geq \sum_{(i,j)\in\mathcal{H}_1^{(0)}} 1_{ \{|Z_{(i,j)}|>\sqrt{2\log|\mathcal{S}|} \} } 
    \geq |\mathcal{H}_1^{(0)}|(1-o(1))
    \geq c_3\sqrt{\log\log(|\mathcal{S}|)}
\end{align*}
for some constant $c_3>0$, where the last inequality holds due to (c) in Assumption \ref{ass:FDR_power}. By using $\bar{\Phi}(\zeta) \approx \frac{1}{\sqrt{2\pi\tau^2}}e^{-\zeta^2/2}$,
\begin{displaymath}
    2\bar{\Phi}(\sqrt{2\log|\mathcal{S}|}) 
    \approx \frac{1}{\sqrt{\pi\log|\mathcal{S}|}}\frac{1}{|\mathcal{S}|}
\end{displaymath}
with $\hat{\zeta}=\sqrt{2\log|\mathcal{S}|}$. Hence, 
\begin{displaymath}
    \frac{2\bar{\Phi}(\sqrt{2\log|\mathcal{S}|}) |\mathcal{S}|}{\max\{1, \sum_{(i,j)\in\mathcal{S}} 1_{\{|Z_{(i,j)}| > \sqrt{2\log(|\mathcal{S}|)}\}}\}} \leq \frac{1}{c_3\sqrt{\pi\log(|\mathcal{S}|)\log\log(|\mathcal{S}|)}} \to 0
\end{displaymath}
as $|\mathcal{S}|\to\infty$. Therefore, $\mathbb{P}(\hat{\zeta}\in(0,\sqrt{2\log|\mathcal{S}|}])\to1$. It remains to show that the proportion of null is uniform over $\hat{\zeta}$. Note that for $(i,j)\in\mathcal{H}_{0}^{(0)}$, since $Z_{(i,j)}$ also asymptotically follows the standard normal distribution, by using Lemma 6.3 in \cite{liu2013gaussian_supp}, we have
\begin{equation}\label{e:glivenko-cantelli}
    \sup_{\hat{\zeta}\in(0,\sqrt{2\log(|\mathcal{S}|}]} \left| \frac{\sum_{(i,j)\in \mathcal{H}_0^{(0)}} 1_{\{|Z_{(i,j)}| > \hat{\zeta}\}}}{|\mathcal{H}_0^{(0)}|2\bar{\Phi}(\hat{\zeta})} -1 \right| \stackrel{p}{\to} 0.
\end{equation}
From the definition of the FDP \eqref{e:fdp}, we have
\begin{displaymath}
    \mathrm{FDP}(\hat\zeta) 
    = \frac{\sum_{(i,j)\in\mathcal{H}_0^{(0)}} 1_{\{|Z_{(i,j)}|>\hat\zeta\}} }{\max\{1,\sum_{(i,j)\in\mathcal{S}} 1_{\{|Z_{(i,j)}|>\hat\zeta\}}\}} \\
    = \frac{|\mathcal{H}_0^{(0)}|}{|\mathcal{S}|} \frac{2\bar\Phi(\hat\tau)|\mathcal{S}|}{\max\{1,\sum_{(i,j)\in\mathcal{S}}1_{\{|Z_{(i,j)}|>\hat\zeta\}}\}}  \frac{\sum_{(i,j)\in\mathcal{H}_0^{(0)}} 1_{\{|Z_{(i,j)}|>\hat\zeta\}}}{|\mathcal{H}_0^{(0)}| 2\bar\Phi(\hat\zeta)}.
\end{displaymath}
Therefore, we obtain the desired result.

Next, consider the test of homogeneity so that take $\mathcal{H}_{0}=\mathcal{H}_{0}^{[K]}$ and $\mathcal{H}_{1}=\mathcal{H}_{1}^{[K]}$. Recall that
\begin{displaymath}
    Z_{(i,j)} 
    = \Phi^{-1}\left(\mathcal{F}_0\left(\sum_{1\leq k<\ell\leq K} (Z_{(i,j)}^{(k,\ell)}(0))^2 \right)\right),
\end{displaymath}
where
\begin{displaymath}
    Z_{(i,j)}^{(k,\ell)}\left((\beta_{i}^{(k)})_j-(\beta_{i}^{(\ell)})_j\right)
    = \frac{\left((\tilde{\beta}_{i}^{(k)})_j-(\tilde{\beta}_{i}^{(\ell)})_j \right) - \left((\beta_{i}^{(k)})_j-(\beta_{i}^{(\ell)})_j \right)}{\sqrt{\frac{1}{N_{k}}\hat{V}_{i,j}^{(k)} + \frac{1}{N_{\ell}}\hat{V}_{i,j}^{(\ell)}}}.
\end{displaymath}
From \eqref{e:clt_beta_pair} in Lemma \ref{lem:clt_beta_pair}, we have
\begin{displaymath}
    Z_{(i,j)}^{(k,\ell)}(0) 
    = Z_{(i,j)}^{(k,\ell)}\left((\beta_{i}^{(k)})_j - (\beta_{i}^{(\ell)})_j\right) + \frac{(\beta_{i}^{(k)})_j - (\beta_{i}^{(\ell)})_j}{\sqrt{\frac{1}{N_k}\hat{V}_{i,j}^{(k)}+\frac{1}{N_{\ell}}\hat{V}_{i,j}^{(\ell)}}}
    := \mathcal{Z}_{(i,j)}^{(k,\ell)} + r_{(i,j)}^{(k,\ell)}.
\end{displaymath}
Note that $\mathcal{Z}_{(i,j)}^{(k,\ell)}\stackrel{d}{\to}\mathcal{N}(0,1)$ by Lemma \ref{lem:clt_beta_pair}. With (a) in Assumption \ref{ass:FDR_power}, if at least one pair of $|r_{(i,j)}^{(k,\ell)}|$ diverges at rate $\sqrt{2\log|\mathcal{S}|}$, then $\chi_{(i,j)}^2\to\infty$. Since $\mathcal{F}_0(x)$ is the CDF, $\mathcal{F}_0(\chi_{(i,j)}^2)\to1$ so that $Z_{(i,j)}\to\infty$. Therefore, $\mathbb{P}(Z_{(i,j)} > \hat{\zeta})\to1$ for any $\hat{\zeta}\leq\sqrt{2\log|\mathcal{S}|}$, and the power is
\begin{displaymath}
    \textrm{Power}(\hat{\zeta}) 
    = \sum_{(i,j)\in \mathcal{H}_{1}^{[K]}}\frac{1_{\{|Z_{(i,j)}| > \hat{\zeta}\}}}{|\mathcal{H}_{1}^{[K]}|} 
    \geq 1 - \max_{(i,j)\in\mathcal{S}}\mathbb{P}(Z_{(i,j)} \leq \hat{\zeta}) \stackrel{p}{\to}1.
\end{displaymath}
Similar to the argument of the test of significance, we can decompose the FDP as
\begin{displaymath}
    FDP(\hat{\zeta}) 
    = \frac{|\mathcal{H}_{0}^{[K]}|}{|\mathcal{S}|} \frac{2\bar{\Phi}(\hat{\zeta})|\mathcal{S}|}{\max\{1,\sum_{(i,j)\in\mathcal{S}}1_{\{|Z_{(i,j)}|>\hat{\zeta}\}}\}}\frac{\sum_{(i,j)\in\mathcal{H}_0^{[K]}} 1_{\{|Z_{(i,j)}|>\hat\zeta\}}}{|\mathcal{H}_0^{[K]}| 2\bar\Phi(\hat\zeta)}.
\end{displaymath}
Note that with (c) in Assumption \ref{ass:FDR_power}, the second factor converges to 0 as $\mathbb{P}(|Z_{(i,j)}|>\sqrt{2\log|\mathcal{S}|}) \to 1$ for $(i,j)\in\mathcal{H}_1^{[K]}$, and the third factor converges to 1 by \eqref{e:glivenko-cantelli}. Therefore, we obtain the desired result.

Finally, consider the test of nullity so that take $\mathcal{H}_{0}=\mathcal{H}_{0}^{[K]_{\circ}}$ and $\mathcal{H}_{1}=\mathcal{H}_{1}^{[K]_{\circ}}$. From \eqref{e:clt_betak}, we have
\begin{displaymath}
    Z_{(i,j)}^{(k)}(0) := \frac{\sqrt{N_k}((\tilde{\beta}_{i}^{(k)})_j - (\beta_{i}^{(k)})_j)}{\sqrt{V_{i,j}^{(k)}}} + \frac{\sqrt{N_k}(\beta_{i}^{(k)})_j}{\sqrt{V_{i,j}^{(k)}}}
    := Z_{(i,j)}^{(k)} + r_{(i,j)}^{(k)},
\end{displaymath}
and
\begin{displaymath}
    Z_{(i,j)} 
    = \Phi^{-1}\left(\mathcal{F}_K\left(\sum_{k=1}^K (Z_{(i,j)}^{(k)}(0))^2 \right)\right).
\end{displaymath}
Note that $Z_{(i,j)}^{(k)} \to \mathcal{N}(0,1)$. With (a) in Assumption \ref{ass:FDR_power}, if at least one pair of $|r_{(i,j)}^{(k,\ell)}|$ diverges at rate $\sqrt{2\log|\mathcal{S}|}$, then $\chi_{(i,j)}^2\to\infty$, so that $Z_{(i,j)}=\Phi^{-1}(\mathcal{F}_K(\chi_{(i,j)}^2))\to\infty$. The remainder of the proof is analogous to that of the test of homogeneity. \qed


\section{Proofs of Lemmas}\label{se:lemmas}

In this section, we provide technical Lemmas and their proofs.

\begin{lemma}[Lemma 11 in \cite{lee2017communication_supp}]\label{lem:lee}
If $\|\tilde{\theta} - \theta\|_{\infty}<\delta_{\theta}$, then for $\hat{\theta}=HT_{\delta_{\theta}}(\tilde{\theta})$, where $HT_{\delta_{\theta}}(\cdot)$ is defined in \eqref{e:hard_thres}, then
\begin{displaymath}
    \|\hat{\theta} - \theta\|_{2} < 2\sqrt{2 s_{\theta}}\delta_{\theta},
\end{displaymath}
where $s_{\theta}=\|\theta\|_0$, the sparsity level of $\theta$.
\end{lemma}

\begin{lemma}\label{lem:minimizer}
    Under the Assumption \ref{ass:aggregation},
    \begin{displaymath}
        (\alpha_{i}^{(0)})_j = \argmin_{x\in\mathbb{R}}\left\{\sum_{k=1}^K\Psi_{\eta_{i,j}}((\beta_{i}^{(k)})_j-x) \right\}
    \end{displaymath}
    is unique and $(\alpha_{i}^{(0)})_j=\mu_{(i,j)}$, where $\mu_{(i,j)}=\frac{1}{|\mathcal{I}_{(i,j)}|}\sum_{k\in \mathcal{I}_{(i,j)}}(\beta_{i}^{(k)})_j$.
\end{lemma}

\begin{proof}
    For $x\in\mathbb{R}$, define $\mathcal{I}_{x}=\{k\in[K]:(\beta_{i}^{(k)})_j\in[x-\eta_{i,j},x+\eta_{i,j}]\}$. First, consider that $\mathcal{I}_x=\mathcal{I}_{(i,j)}$. By condition (b) in Assumption \ref{ass:aggregation}, $|(\beta_{i}^{(k)})_j-\mu_{(i,j)}|\leq \delta_{\textrm{in}}$ for $k\in\mathcal{I}_{(i,j)}$. By condition (c), $\eta_{i,j}>2\delta_{\textrm{in}}$, so all inliers lie within $\eta_{i,j}$ of $\mu_{(i,j)}$. Moreover, for any $k\notin\mathcal{I}_{(i,j)}$, 
    condition (c) gives $|(\beta_i^{(k)})_j-(\beta_i^{(k_1)})_j|\geq\delta_{\min}>2\eta_{i,j}$ 
    for any $k_1\in\mathcal{I}_{(i,j)}$, so no outlier falls within $\eta_{i,j}$ of $\mu_{(i,j)}$. Hence $\mathcal{I}_{\mu_{(i,j)}}=\mathcal{I}_{(i,j)}$, and
    \begin{displaymath}
        \sum_{k=1}^K\Psi_{\eta_{i,j}}((\beta_{i}^{(k)})_j-x) 
        = (K-|\mathcal{I}_{(i,j)}|)\eta_{i,j}^2 + \sum_{k\in \mathcal{I}_{(i,j)}}((\beta_{i}^{(k)})_j-x)^2,
    \end{displaymath}
    which is minimized uniquely at $x=\mu_{(i,j)}$, including the case $\delta_{\textrm{in}}=0$.
    
    Next, consider that $\mathcal{I}_x \neq \mathcal{I}_{(i,j)}$ and $\delta_{\textrm{in}}<|\mu_{(i,j)}-x|\leq\delta_{\textrm{gap}}/2$. For simplicity, suppose $x>(\beta_{i}^{(k)})_j$ for $k\in \mathcal{I}_{(i,j)}$, and let $(\beta_{i}^{\max})_j=\max_{k\in \mathcal{I}_{(i,j)}}(\beta_{i}^{(k)})_j$. Since the inlier diameter is at most $2\delta_{\textrm{in}}$ and $\eta_{i,j}>2\delta_{\textrm{in}}$, we have
    \begin{displaymath}
        ((\beta_{i}^{\max})_j - (\beta_{i}^{(k)})_j) 
        \leq 2\delta_{\textrm{in}} < \eta_{i,j}
    \end{displaymath}
    for all $k\in\mathcal{I}_{(i,j)}$, which implies $\mathcal{I}_{(\beta_{i}^{\max})_j}=\mathcal{I}_{(i,j)}$. Since $((\beta_{i}^{(k)})_j-x)^2 \geq ((\beta_{i}^{(k)})_j-(\beta_{i}^{\max})_j)^2$ for $k\in\mathcal{I}_{(i,j)}$, we have
    \begin{displaymath}
        \Psi_{\eta_{i,j}}((\beta_{i}^{(k)})_j-x) \geq 
        \Psi_{\eta_{i,j}}((\beta_{i}^{(k)})_j-(\beta_{i}^{\max})_j), 
        \quad k\in\mathcal{I}_{(i,j)}.
    \end{displaymath}
    For $k\notin\mathcal{I}_{(i,j)}$, by the definition of $\delta_{\min}$ in condition (c) of Assumption \ref{ass:aggregation}, $|(\beta_i^{(k)})_j - (\beta_i^{(k_1)})_j|\geq\delta_{\min}$ for any $k_1\in\mathcal{I}_{(i,j)}$. Taking $k_1$ to achieve $(\beta_i^{\max})_j$ gives $|(\beta_i^{(k)})_j -(\beta_i^{\max})_j|\geq\delta_{\min}>2\eta_{i,j}$ by condition (c), 
    so $\Psi_{\eta_{i,j}}((\beta_{i}^{(k)})_j-(\beta_i^{\max})_j)=\eta_{i,j}^2$. This holds at $x=(\beta_i^{\max})_j$ as well. Hence,
    \begin{displaymath}
        \sum_{k=1}^K\Psi_{\eta_{i,j}}((\beta_{i}^{(k)})_j -x) 
        \geq \sum_{k=1}^K\Psi_{\eta_{i,j}}((\beta_{i}^{(k)})_j -(\beta_{i}^{\max})_j)
        \geq \sum_{k=1}^K\Psi_{\eta_{i,j}}((\beta_{i}^{(k)})_j -\mu_{(i,j)}),
    \end{displaymath}
    so the minimum is at $x=\mu_{(i,j)}$. The case $x<(\beta_{i}^{(k)})_j$ for 
    $k\in \mathcal{I}_{(i,j)}$ follows similarly.
    
    Finally, consider $\mathcal{I}_x \neq \mathcal{I}_{(i,j)}$ and $|\mu_{(i,j)}-x|>\delta_{\textrm{gap}}/2$. In this case $\mathcal{I}_x\cap\mathcal{I}_{(i,j)}=\emptyset$, since for $k\in\mathcal{I}_{(i,j)}$,
    \begin{displaymath}
        |(\beta_{i}^{(k)})_j-x|\geq|\mu_{(i,j)}-x|-\delta_{\textrm{in}}
        >\delta_{\textrm{gap}}/2-\delta_{\textrm{in}}
        =(\delta_{\textrm{out}}-3\delta_{\textrm{in}})/2
        \geq\delta_{\textrm{in}}\geq 0,
    \end{displaymath}
    where the last inequality uses $\delta_{\textrm{out}}\geq 5\delta_{\textrm{in}}$. Since $\eta_{i,j}<\delta_{\min}/2$ by condition (c), we have $|(\beta_{i}^{(k)})_j-x|\geq\eta_{i,j}$ for $k\in\mathcal{I}_{(i,j)}$, which gives us
    \begin{displaymath}
        \sum_{k=1}^K\Psi_{\eta_{i,j}}((\beta_{i}^{(k)})_j-x) 
        \geq |\mathcal{I}_{(i,j)}|\eta_{i,j}^2.
    \end{displaymath}
    On the other hand, since $\eta_{i,j}>2\delta_{\textrm{in}}$ by condition (c),
    \begin{align*}
        \sum_{k=1}^K\Psi_{\eta_{i,j}}((\beta_{i}^{(k)})_j-\mu_{(i,j)})
        &= (K-|\mathcal{I}_{(i,j)}|)\eta_{i,j}^2 
        + \sum_{k\in\mathcal{I}_{(i,j)}}((\beta_{i}^{(k)})_j-\mu_{(i,j)})^2 \\
        &\leq (K-|\mathcal{I}_{(i,j)}|)\eta_{i,j}^2 + |\mathcal{I}_{(i,j)}|\delta_{\textrm{in}}^2 \\
        &< (K-|\mathcal{I}_{(i,j)}|)\eta_{i,j}^2 + |\mathcal{I}_{(i,j)}|\frac{\eta_{i,j}^2}{4} \\
        &= \left(K - \frac{3|\mathcal{I}_{(i,j)}|}{4}\right)\eta_{i,j}^2.
    \end{align*}
    For the minimum to be at $\mu_{(i,j)}$, it suffices to have
    \begin{displaymath}
        \left(K - \frac{3|\mathcal{I}_{(i,j)}|}{4}\right)\eta_{i,j}^2 
        \leq |\mathcal{I}_{(i,j)}|\eta_{i,j}^2,
    \end{displaymath}
    which simplifies to $|\mathcal{I}_{(i,j)}|/K \geq 4/7$, as required by condition (a) of Assumption \ref{ass:aggregation}.
\end{proof}

\begin{lemma}[Proposition 2.2 in \cite{basu2024high_supp}]\label{lem:beta_hat_bdd}
Suppose that (a) -- (c) in Assumption \ref{ass:ts} hold. With a choice of $\lambda^{(k)}=\mathcal{O}_{\mathbb{P}}(\sigma_{k,\max}^2\kappa^2(\Phi^{(k)})\sqrt{\log d/N_k})$, $k\in[K]$, the following holds.
\begin{equation}\label{e:bdd_sparse}
    \|\hat{\beta}_i^{(k)} - \beta_{i}^{(k)}\|_1 = \|\hat{\Phi}_{i:}^{(k)'} - \Phi_{i:}^{(k)'}\|_1
    = \mathcal{O}_{\mathbb{P}}\left(\kappa(\Sigma_{\epsilon}^{(k)})\kappa^2(\Phi^{(k)})|\!|\!|\Phi^{(k)}|\!|\!|^2s_{\Phi}^{(k)}\sqrt{\frac{\log d}{N_k}}\right), \ i\in[d].
\end{equation}
\end{lemma}

\begin{lemma}\label{lem:delta_bdd}
Consider that for $i\in[d]$,
\begin{displaymath}
    \Delta_i^{(k)} 
    := (I-\hat{\Theta}^{(k)}\hat{\Sigma}^{(k)})(\hat{\beta}_{i}^{(k)} - \beta_{i}^{(k)}).
\end{displaymath}
Then 
\begin{equation}\label{e:delta}
    \|\Delta_i^{(k)}\|_{\infty} 
    = \mathcal{O}_{\mathbb{P}}\left(\kappa^2(\Sigma_{\epsilon}^{(k)})\kappa^4(\Phi^{(k)})|\!|\!|\Phi^{(k)}|\!|\!|^2 s_{\Phi}^{(k)}\frac{\log d}{N_k}\right).
\end{equation}
\end{lemma}

\begin{proof}
Recall the debiased equation \eqref{e:debiased}. For for each $i$ and $k$, one has
\begin{equation}\label{e:difference}
    \tilde{\beta}_i^{(k)} - \beta_{i}^{(k)} = \frac{1}{N_k}\hat{\Theta}^{(k)}\mathcal{X}^{(k)'}\underbrace{(\mathcal{Y}_i^{(k)} - \mathcal{X}^{(k)}\beta_{i}^{(k)})}_{=E_{i}^{(k)}} + \underbrace{(I-\hat{\Theta}^{(k)}\hat{\Sigma}^{(k)})(\hat{\beta}_{i}^{(k)} - \beta_{i}^{(k)})}_{=\Delta_i^{(k)}}.
\end{equation}
Note that 
\begin{align*}
    \|\Delta_i^{(k)}\|_{\infty} 
    &\leq \|I - \hat{\Theta}^{(k)}\hat{\Sigma}^{(k)}\|_{\infty}\|\hat{\beta}_i^{(k)} - \beta_{i}^{(k)}\|_1 = \max_{j\in[d]}|e_j - \hat{\Theta}_{j:}^{(k)}\hat{\Sigma}^{(k)}|\|\hat{\beta}_i^{(k)} - \beta_{i}^{(k)}\|_1.
\end{align*}
From KKT condition for nodewise regression in \eqref{e:nodewise},
\begin{displaymath}
    -\frac{1}{N_k}\mathcal{X}_{-j}^{(k)'}(\mathcal{X}_{j}^{(k)}-\mathcal{X}_{-j}^{(k)}\hat{\gamma}_{j}^{(k)}) + \lambda_{j}^{(k)}\hat{z}_{j}^{(k)},
\end{displaymath}
where $\|\hat{z}_{j}^{(k)}\|_1\leq1$. This implies 
\begin{displaymath}
    -\frac{1}{N_k}\hat{\gamma}_{j}^{(k)'}\mathcal{X}_{-j}^{(k)'}\mathcal{X}^{(k)}\hat{\Theta}_{j:}^{(k)'}(\hat{\tau}_j^{(k)})^{2} + \lambda_{j}^{(k)}\|\hat{\gamma}_j^{(k)}\|_1 = 0.
\end{displaymath}
This gives us
\begin{align}
    (\hat{\tau}_j^{(k)})^2 &= \frac{1}{N_k}\|\mathcal{X}_{j}^{(k)} - \mathcal{X}_{-j}^{(k)}\hat{\gamma}_j^{(k)}\|_2^2 + \lambda_j^{(k)}\|\hat{\tau}_j^{(k)}\|_1 \label{e:sample_error_variance}\\
    &= \frac{1}{N_k}\mathcal{X}_{j}^{(k)'}\mathcal{X}^{(k)}\hat{\Theta}_{j:}^{(k)'}(\hat{\tau}_j^{(k)})^2 -\frac{1}{N_k}\hat{\gamma}_{j}^{(k)'}\mathcal{X}_{-j}^{(k)'}\mathcal{X}^{(k)}\hat{\Theta}_{j:}^{(k)'}(\hat{\tau}_j^{(k)})^{2} + \lambda_{j}^{(k)}\|\hat{\gamma}_j^{(k)}\|_1 \nonumber\\
    &= \frac{1}{N_k}\mathcal{X}_{j}^{(k)'}\mathcal{X}^{(k)}\hat{\Theta}_{j:}^{(k)'}(\hat{\tau}_j^{(k)})^2. \nonumber
\end{align}
Hence, $\frac{1}{N_k}\mathcal{X}_{j}^{(k)'}\mathcal{X}^{(k)}\hat{\Theta}_{j:}^{(k)}=1$. This implies 
\begin{displaymath}
    |e_j - \hat{\Theta}_{j:}^{(k)}\hat{\Sigma}^{(k)}| \leq \frac{\lambda_j^{(k)}}{(\hat{\tau}_j^{(k)})^2} \leq \frac{\sigma_{k,\max}^2\kappa^2(\Phi^{(k)})}{(\hat{\tau}_j^{(k)})^2}\sqrt{\frac{\log d}{N_k}}
\end{displaymath}
with a suitable choice of $\lambda_j^{(k)}=\mathcal{O}_{\mathbb{P}}(\sigma_{k,\max}^2\kappa^2(\Phi^{(k)})\sqrt{\log d/N_k})$ uniformly over $j\in[d]$. 

To complete the proof, we define the following two constants, the lower and upper extreme spectral densities \cite[e.g.,][]{basu2015regularized_supp,deb2026inference_supp}; for spectral densities of $\{X_t^{(k)}\}$ denoted as $f_X^{(k)}(\omega)$ at frequency $\omega$,
\begin{displaymath}
    m(f_X^{(k)}) 
    := \underset{\omega\in[-\pi,\pi]}{\textrm{ess inf}} \ \Lambda_{\min}(f_X^{(k)}(\omega)) \quad
    \mathcal{M}(f_X^{(k)}) 
    := \underset{\omega\in[-\pi,\pi]}{\textrm{ess sup}} \ \Lambda_{\max}((f_X^{(k)}(\omega))^{\dagger}f_X^{(k)}(\omega)),
\end{displaymath}
where $\Lambda_{\min}(M),\Lambda_{\max}(M)$ are the minimum and maximum eigenvalue of the matrix $M$, respectively, and $\dagger$ refers to the conjugate transpose. Then we reconstruct Lemma 5.3 in \cite{van2014asymptotically_supp} using this notation: the population error variance $(\tau_j^{(k)})^{2} = \mathbb{E}[(\mathcal{X}_{1,j} - \sum_{\ell \neq j}\gamma_{j,\ell}X_{1,\ell})^2]$ satisfies $(\tau_j^{(k)})^{2} = \mathbb{E}\epsilon_{1,j}^2 = \frac{1}{\Theta_{jj}^{(k)}} \geq \sigma_{k,\min}^2 >0$ and $(\tau_j^{(k)})^{2} \leq \mathcal{M}(f_X^{(k)}) \leq \mathcal{O}(\sigma_{k,\max}^2|\!|\!|(\Phi^{(k)})^{-1}|\!|\!|^2) < \infty$. This successfully replaces Assumption (A2) in \cite{van2014asymptotically_supp}. From \eqref{e:sample_error_variance},
\begin{align}
    \frac{1}{N_k}\|\mathcal{X}_{j}^{(k)} - \mathcal{X}_{-j}^{(k)}\hat{\gamma}_j^{(k)}\|_2^2
    &= \frac{1}{N_k}\|\mathcal{X}_{j}^{(k)} - \mathcal{X}_{-j}^{(k)}\gamma_j^{(k)}\|_2^2 \nonumber\\
    & + \frac{1}{N_k}\|\mathcal{X}_{-j}(\hat{\gamma}_j^{(k)}-\gamma_j^{(k)})\|_2^2 \nonumber\\
    & + \frac{2}{N_k}(\mathcal{X}_{j}^{(k)} - \mathcal{X}_{-j}^{(k)}\gamma_j^{(k)})'\mathcal{X}_{-j}^{(k)}(\hat{\gamma}_j^{(k)}-\gamma_j^{(k)}). \label{e:sample_error_variance2}
\end{align}
The first term in \eqref{e:sample_error_variance2} is $(\tau_j^{(k)})^2$. By Proposition 3.3 in \cite{basu2015regularized_supp}, the upper bound of the second term in \eqref{e:sample_error_variance2} is
\begin{displaymath}
    \mathcal{O}_{\mathbb{P}}\left( \frac{(s_{\Theta}^{(k)})_j(\lambda_j^{(k)}(\mathcal{M}(f_X^{(k)})+\mathcal{M}(f_{\epsilon}^{(k)}))^2}{m(f_X^{(k)})} \right) 
    = \mathcal{O}_{\mathbb{P}}\left(\kappa^2(\Sigma_{\epsilon}^{(k)})\kappa^4(\Phi^{(k)})|\!|\!|\Phi^{(k)}|\!|\!|^2\frac{(s_{\Theta}^{(k)})_j\log d}{N_k}\right) = \mathcal{O}_{\mathbb{P}}\left(\frac{1}{\sqrt{N_k}}\right),
\end{displaymath}
where the last equality holds by Assumption \ref{ass:ts}. By Propositions 3.2 and 3.3 in \cite{basu2015regularized_supp}, the third term in \eqref{e:sample_error_variance2} is bounded above by
\begin{align*}
    & 2\left\|\frac{1}{N_k}E_{-j}^{(k)'}\mathcal{X}_{-j}^{(k)}\right\|_{\infty}\|\hat{\gamma}_j^{(k)}-\gamma_j^{(k)}\|_1 \\
    &= \mathcal{O}_{\mathbb{P}}\left( \lambda_j^{(k)}(\mathcal{M}(f_X^{(k)})+\mathcal{M}(f_{\epsilon}^{(k)})) \right)\mathcal{O}_{\mathbb{P}}\left( \frac{(s_{\Theta}^{(k)})_j(\mathcal{M}(f_X^{(k)})+\mathcal{M}(f_{\epsilon}^{(k)}))}{m(f_X^{(k)})}\sqrt{\frac{\log d}{N_k}}\right) \\
    &= \mathcal{O}_{\mathbb{P}}\left(\kappa(\Sigma_{\epsilon}^{(k)})\kappa^2(\Phi^{(k)})\sqrt{\frac{\log d}{N_k}}\right)\mathcal{O}_{\mathbb{P}}\left(\kappa(\Sigma_{\epsilon}^{(k)})\kappa^2(\Phi^{(k)})|\!|\!|\Phi^{(k)}|\!|\!|^2 (s_{\Theta}^{(k)})_j\sqrt{\frac{\log d}{N_k}}\right) \\
    &= \mathcal{O}_{\mathbb{P}}\left(\kappa^2(\Sigma_{\epsilon}^{(k)})\kappa^4(\Phi^{(k)})|\!|\!|\Phi^{(k)}|\!|\!|^2\frac{(s_{\Theta}^{(k)})_j\log d}{N_k}\right) \\
    &= \mathcal{O}_{\mathbb{P}}\left(\frac{1}{\sqrt{N_k}}\right).
\end{align*}
In addition, note that $\|\gamma_j^{(k)}\|_1\leq\sqrt{(s_{\Theta}^{(k)})_j}\|\gamma_j^{(k)}\|_2\leq \mathcal{O}(\sqrt{(s_{\Theta}^{(k)})_j})$. So, with a choice of $\lambda_j^{(k)}\asymp\sqrt{\log d/N_k}$ over $j\in[d]$,
\begin{equation}\label{e:sample_error_variance3}
    \lambda_{j}^{(k)}\|\hat{\gamma}_j^{(k)}\|_1 
    \leq \lambda_{j}^{(k)}\|\gamma_j^{(k)}\|_1 + \lambda_{j}^{(k)}\|\hat{\gamma}_j^{(k)} - \gamma_j^{(k)}\|_1 
    = \lambda_{j}^{(k)}\mathcal{O}(\sqrt{(s_{\Theta}^{(k)})_j}) + \lambda_{j}^{(k)}\mathcal{O}_{\mathbb{P}}(\lambda_{j}^{(k)}(s_{\Theta}^{(k)})_j) = o_{\mathbb{P}}(1).
\end{equation}
Combining \eqref{e:sample_error_variance2}, \eqref{e:sample_error_variance3}, and $(\tau_j^{(k)})^{2} \geq \sigma_{k,\min}^2$ yields 
\begin{equation}\label{e:max_tau_hat}
    \max_{j\in[d]}\frac{1}{(\hat{\tau}_j^{(k)})^2} = \mathcal{O}_{\mathbb{P}}(1/\sigma_{k,\min}^2).
\end{equation}
This implies 
\begin{equation}\label{e:bias_covariance}
        \|I - \hat{\Theta}^{(k)}\hat{\Sigma}^{(k)}\|_{\infty} 
        =  \mathcal{O}_{\mathbb{P}}\left(\kappa(\Sigma_{\epsilon}^{(k)})\kappa^2(\Phi^{(k)})\sqrt{\frac{\log d}{N_k}}\right).
\end{equation}
Combining with \eqref{e:bdd_sparse} and \eqref{e:bias_covariance} yields the desired result \eqref{e:delta}.    
\end{proof}

\begin{lemma}\label{lem:beta_tilde_bdd}
The bound on $\tilde{\beta}_i^{(k)}-\beta_{i}^{(k)}$, $i\in[d]$, in \eqref{e:difference} is
\begin{equation}\label{e:beta_tilde_bdd}
    \|\tilde{\beta}_i^{(k)}-\beta_{i}^{(k)}\|_{\infty} 
    = \mathcal{O}_{\mathbb{P}}\left(\kappa(\Sigma_{\varepsilon}^{(k)})\sqrt{\frac{\log d}{N_{k}}}\right).
\end{equation}
Consequently, taking the union bound over all $i\in[d]$ yields the uniform bound:
\begin{displaymath}
    \max_{i \in [d]}\|\tilde{\beta}i^{(k)} - \beta_i^{(k)}|_{\infty}
    = \mathcal{O}_{\mathbb{P}}\left(\kappa(\Sigma\varepsilon^{(k)})\sqrt{\frac{\log d^2}{N_k}}\right).
\end{displaymath}
\end{lemma}

\begin{proof}
Note that
\begin{equation}\label{e:bdd_debiased}
    \|\tilde{\beta}_i^{(k)}-\beta_{i}^{(k)}\|_{\infty} \leq \left\|\frac{1}{N_k}\Theta^{(k)}\mathcal{X}^{(k)'}E_{i}^{(k)}\right\|_{\infty} + \left\|(\hat{\Theta}^{(k)}-\Theta^{(k)})\frac{1}{N_k}\mathcal{X}^{(k)'}E_{i}^{(k)} \right\|_{\infty} + \|\Delta_i^{(k)}\|_{\infty}.
\end{equation}
From Lemma \ref{lem:delta_bdd}, the third term of \eqref{e:bdd_debiased} is bounded above by \eqref{e:delta}. Note that the first term of \eqref{e:bdd_debiased} is 
\begin{displaymath}
    \max_{j\in[d]}\frac{1}{\sqrt{N_k}}\left|\frac{1}{\sqrt{N_k}}e_j\Theta^{(k)}\mathcal{X}^{(k)'}E_{i}^{(k)}\right|.
\end{displaymath}
By Lemma \eqref{lem:martingale_clt}, we have
\begin{displaymath}
    \frac{1}{\sqrt{N_k}}\sum_{t=p+1}^{T_k}(\Theta_{j:}^{(k)'}X_{t-1}^{(k)})\varepsilon_{i,t} ^{(k)}
    := S_j^{(k)}
    \stackrel{d}{\to} \mathcal{N}(0,\sigma_{k,i}^2\Theta_{j,j}^{(k)}).
\end{displaymath}
Since $\Theta_{j,j}^{(k)} \leq 1/\sigma_{k,\min}^2$, the variance is bounded above by $\sigma_{k,i}^2/\sigma_{k,\min}^2\leq \kappa^2(\Sigma_{\varepsilon}^{(k)})$. Applying Borell–TIS inequality \cite[e.g., Theorem 2.1.1 in][]{adler2007gaussian_supp}  with $u=\log d$, we obtain
\begin{displaymath}
    \mathbb{P}\left(\max_{j\in[d]}|S_j^{(k)} | - \mathbb{E}\max_{j\in[d]}|S_j^{(k)} | > \sqrt{\log d} \right)
    \leq d^{-1/2}.
\end{displaymath}
By the maximum bound for sub-Gaussian random variables \cite[e.g., Chapter 2.5.2 in][]{vershynin2018high_supp}, $\mathbb{E}\max_{j\in[d]}|S_j^{(k)} | \leq \kappa^2(\Sigma_{\varepsilon}^{(k)})\sqrt{2\log d}$. Therefore, dividing by $\sqrt{N_k}$ yields the upper bound of the first term in \eqref{e:bdd_debiased} as 
\begin{displaymath}
    \mathcal{O}_{\mathbb{P}}\left(\kappa(\Sigma^{(k)})\sqrt{\frac{\log d}{N_k}}\right).
\end{displaymath}

Note that the second term of \eqref{e:bdd_debiased} is 
\begin{align}
    \left\|(\hat{\Theta}^{(k)}-\Theta^{(k)})\frac{1}{N_k}\mathcal{X}^{(k)'}E_{i}^{(k)}\right\|_{\infty} 
    &= \max_{j\in[d]}\left|(\hat{\Theta}_{j:}^{(k)}-\Theta^{(k)}_{j:})\frac{1}{N_k}\mathcal{X}^{(k)'}E_{i}^{(k)}\right| \nonumber\\
    &\leq \max_{j\in[d]}\|\hat{\Theta}_{j:}^{(k)}-\Theta^{(k)}_{j:}\|_1\left\|\frac{1}{N_k}\mathcal{X}^{(k)'}E_{i}^{(k)}\right\|_{\infty} \nonumber\\
    &\leq \max_{j\in[d]}\|\hat{\Theta}_{j:}^{(k)}-\Theta^{(k)}_{j:}\|_1\mathcal{O}_{\mathbb{P}}\left(\sigma_{k,\max}^2\kappa^2(\Phi^{(k)})\sqrt{\frac{\log d}{N_k}}\right) \nonumber\\
    &= \mathcal{O}_{\mathbb{P}}\left(\kappa^2(\Sigma_{\epsilon}^{(k)})\kappa^4(\Phi^{(k)})|\!|\!|\Phi^{(k)}|\!|\!|^2s_{\Theta}^{(k)}\frac{\log d}{N_k}\right) \label{e:second_term}
\end{align}
where the last equality holds since, by the definition of the deviation bound \cite[e.g., Section 4 in][]{basu2015regularized_supp},
\begin{equation}\label{e:deviation_bound}
    \left\|\frac{1}{N_k}\mathcal{X}^{(k)'}E_{i}^{(k)}\right\|_{\infty}  
    \leq \mathcal{O}_{\mathbb{P}}(\sigma_{k,\max}^2\kappa^2(\Phi^{(k)}))\sqrt{\frac{\log d}{N_k}},
\end{equation}
which gives the result. To show the last inequality in \eqref{e:second_term}, note that from $\hat{\Theta}^{(k)} = (\hat{\gamma}^{(k)})^{-2}\hat{\Gamma}^{(k)}$, so for fixed $j$, $\hat{\Theta}_{j:}^{(k)} = \hat{\gamma}_{j}^{(k)}/(\hat{\tau}_{j}^{(k)})^2$. Then we have
\begin{align}
    & \|\hat{\Theta}_{j:}^{(k)} - \Theta_{j:}^{(k)}\|_1 \nonumber\\
    &= \|\hat{\gamma}_{j}^{(k)}/(\hat{\tau}_{j}^{(k)})^2 - \gamma_{j}^{(k)}/(\tau_{j}^{(k)})^2 \|_1 \nonumber\\
    &\leq \|\hat{\gamma}_{j}^{(k)} - \gamma_{j}^{(k)}\|_1/(\hat{\tau}_j^{(k)})^2 + \|\gamma_{j}^{(k)}\|_1(1/(\hat{\tau}_j^{(k)})^2 - 1/(\tau_j^{(k)})^2). \label{e:theta_difference_ell1}
\end{align}
Consider the second term. Recall that the definition of $(\tau_j^{(k)})^2$ is $\frac{1}{N_k}\|\mathcal{X}_j^{(k)}-\mathcal{X}_{-j}^{(k)}\gamma_j^{(k)}\|_2^2$. Then by using \eqref{e:sample_error_variance2} and \eqref{e:sample_error_variance3}, its upper bound is $o_{\mathbb{P}}(1)$. The upper bound of the first term is 
\begin{equation}\label{e:bound_theta}
    \mathcal{O}_{\mathbb{P}}\left(\kappa(\Sigma_{\epsilon}^{(k)})\kappa^2(\Phi^{(k)})|\!|\!|\Phi^{(k)}|\!|\!|^2(s_{\Theta}^{(k)})_j\sqrt{\frac{\log d}{N_k}}\right)\mathcal{O}_{\mathbb{P}}(1/\sigma_{k,\min}^2),
\end{equation}
so that the upper bound of $\|\hat{\Theta}_{j:}^{(k)} - \Theta_{j:}^{(k)}\|_1$ is \eqref{e:bound_theta}. Hence, combining the results above yields the upper bound of \eqref{e:bdd_debiased},
\begin{displaymath}
     \|\tilde{\beta}_i^{(k)}-\beta_{i}^{(k)}\|_{\infty} 
     = \mathcal{O}_{\mathbb{P}}\left(\kappa(\Sigma_{\varepsilon}^{(k)})\sqrt{\frac{\log d}{N_k}} + \kappa^2(\Sigma_{\epsilon}^{(k)})\kappa^4(\Phi^{(k)})|\!|\!|\Phi^{(k)}|\!|\!|^2\max\{s_{\Theta}^{(k)},s_{\Phi}^{(k)}\}\frac{\log d}{N_k}\right).
    \end{displaymath}
By Assumptions \ref{ass:ts}, we have the desired result \eqref{e:beta_tilde_bdd}.
\end{proof}

\begin{lemma}\label{lem:alpha_bdd}
If Assumptions \ref{ass:ts} and \ref{ass:aggregation} are satisfied, we have 
\begin{displaymath}
    \|\tilde{\alpha}^{(0)} - \alpha^{(0)}\|_{\infty} 
    = \mathcal{O}_{\mathbb{P}}\left(\max_{k\in[K]}\kappa(\Sigma_{\epsilon}^{(k)})\sqrt{\frac{\log d^2}{(1-e)KN_{\min}}}\right)
\end{displaymath}
for $e\in(0,3/7]$ and for all $k\in[K]$,
\begin{displaymath}
    \|\tilde{\alpha}^{(k)} - \alpha^{(k)}\|_{\infty} 
    = \mathcal{O}_{\mathbb{P}}\left(\kappa(\Sigma_{\epsilon}^{(k)})\sqrt{\frac{\log d^2}{N_k}}\right).
\end{displaymath}
\end{lemma}

\begin{proof}
From Lemma \ref{lem:minimizer}, the objective function \eqref{e:redescending_estim} has a unique minimizer $(\alpha_{i}^{(0)})_j = \mu_{(i,j)}$. Note that if $\kappa(\Sigma^{(k)})\sqrt{\log d/N_{k}}\leq \frac{1}{4}(\eta_{i,j}-2\delta_{\textrm{in}})\wedge(\eta_{i,j} -\delta_{\min}/2)$ for all $k$, there exist $\hat{\delta}_{\textrm{in}},\hat{\delta}_{\min}$ such that
\begin{displaymath}
    2\hat{\delta}_{\textrm{in}} < \eta_{i,j} < \hat{\delta}_{\min}/2,
\end{displaymath}
and Assumption \ref{ass:aggregation} holds with $\hat{\delta}_{\textrm{in}}$ and $\hat{\delta}_{\min}$ in place of $\delta_{\textrm{in}}$ and $\delta_{\min}$. Hence, we have $(\tilde{\alpha}_{i}^{(0)})_j= \frac{1}{|\mathcal{I}_{(i,j)}|}\sum_{k\in \mathcal{I}_{(i,j)}}(\tilde{\beta}_i^{(k)})_j$.

Define $w_{ij}^{(k)} = \frac{1}{|\mathcal{I}_{(i,j)}|}1_{\{\mathcal{I}_{(i,j)}\}}(k)$, the indicator function of the events $k\in \mathcal{I}_{(i,j)}$. Note that 
\begin{align*}
    (\tilde{\alpha}_{i}^{(0)})_j
    &= \sum_{k=1}^K w_{ij}^{(k)}(\tilde{\beta}_{i}^{(k)})_j \\
    &= \sum_{k=1}^K w_{ij}^{(k)}\left((\beta_{i}^{(k)})_j - \frac{1}{N_k}\hat{\Theta}_{j:}^{(k)}\mathcal{X}^{(k)'}E_{i}^{(k)} + \Delta_{i,j}^{(k)} \right) \\
    &= (\alpha_{i}^{(0)})_j - \sum_{k=1}^K\frac{w_{ij}^{(k)}}{N_k}\hat{\Theta}_{j:}^{(k)}\mathcal{X}^{(k)'}E_{i}^{(k)} + \sum_{k=1}^K w_{ij}^{(k)}\Delta_{i,j}^{(k)}.
\end{align*}
Then we have
\begin{align}
    & |(\tilde{\alpha}_{i}^{(0)})_j- (\alpha_{i}^{(0)})_j| \nonumber\\
    &\leq \left| \sum_{k=1}^K\frac{w_{ij}^{(k)}}{N_k}\hat{\Theta}_{j:}^{(k)}\mathcal{X}^{(k)'}E_{i}^{(k)} \right| + \left| \sum_{k=1}^K w_{ij}^{(k)}\Delta_{i,j}^{(k)} \right| \nonumber\\
    &\leq \left| \sum_{k=1}^K\frac{w_{ij}^{(k)}}{N_k}\Theta_{j:}^{(k)}\mathcal{X}^{(k)'}E_{i}^{(k)} \right| + \left| \sum_{k=1}^K\frac{w_{ij}^{(k)}}{N_k}(\hat{\Theta}_{j:}^{(k)} - \Theta_{j:}^{(k)})\mathcal{X}^{(k)'}E_{i}^{(k)} \right|  + \left| \sum_{k=1}^K w_{ij}^{(k)}\Delta_{i,j}^{(k)} \right| \nonumber\\
    &\leq \left| \sum_{k=1}^K\frac{w_{ij}^{(k)}}{N_k}\Theta_{j:}^{(k)}\mathcal{X}^{(k)'}E_{i}^{(k)} \right| + \sum_{k=1}^K w_{ij}^{(k)}\left( \|\hat{\Theta}_{j:}^{(k)} - \Theta_{j:}^{(k)}\|_1 \left\|\frac{1}{N_k}\mathcal{X}^{(k)'}E_{i}^{(k)} \right\|_{\infty} + |\Delta_{i,j}^{(k)}|\right). \label{e:individual_sum}
\end{align}
From \eqref{e:bound_theta}, \eqref{e:deviation_bound}, and \eqref{e:delta}, the second term in \eqref{e:individual_sum} is bounded above by
\begin{displaymath}
    \mathcal{O}_{\mathbb{P}}\left(\sum_{k=1}^K w_{i,j}^{(k)} \left(\kappa^2(\Sigma_{\epsilon}^{(k)})\kappa^4(\Phi^{(k)})|\!|\!|\Phi^{(k)}|\!|\!|^2\max\{s_{\Theta}^{(k)},s_{\Phi}^{(k)}\}\frac{\log d}{N_{k}}\right)\right).
\end{displaymath}
Note that Assumption \ref{ass:aggregation} prevents $\{w_{i,j}^{(k)}\}$ from being degenerated, and by Assumption \ref{ass:ts}, the convergence are uniform over all $i\in[d],j\in[d]$. Hence, the union bound over $i,j$, and $k$ regarding the second term in \eqref{e:individual_sum} becomes negligible. 

Consider the first term in \eqref{e:individual_sum}. For each $k\in[K]$, by Lemma \eqref{lem:martingale_clt},
\begin{displaymath}
   \frac{1}{\sqrt{N_k}}\sum_{t=1}^{N_k}(\Theta_{j:}^{(k)'}X_{t-1}^{(k)})\varepsilon_{i,t}^{(k)} 
   := S_{N_k}^{(k)} 
   \stackrel{d}{\to} \mathcal{N}(0,\sigma_{k,i}^2\Theta_{j,j}^{(k)}).
\end{displaymath}
Therefore, by using generalized Hoeffding inequality \cite[e.g., Theorem 2.6.3 in][]{vershynin2018high_supp}, with $a=(w_{i,j}^{(1)}/\sqrt{N_1},\ldots,w_{i,j}^{(K)}/\sqrt{N_k})$,
\begin{displaymath}
    \mathbb{P}\left( \left| \sum_{k=1}^K\frac{1}{\sqrt{N_k}}S_{N_k}^{(k)} \right| \geq \eta \right) \leq 2\exp\left( -\frac{c_1\eta^2}{-c_2\max_{k\in[K]}(\sigma_{k,i}^2\Theta_{j,j}^{(k)})\|a\|_2^2} \right).
\end{displaymath}
for some positive constants $c_1$ and $c_2$. By taking $\eta=\mathcal{O}(\max_{k\in[K]}\kappa(\Sigma_{\epsilon}^{(k)})\sqrt{\log d^2}\|a\|_2)$ and using $\Theta_{j,j}^{(k)} \leq 1/\sigma_{k,\min}^2$, the first term in \eqref{e:individual_sum} is bounded above by
\begin{displaymath}
    \mathcal{O}_{\mathbb{P}}\left(\max_{k\in[K]}\kappa(\Sigma_{\epsilon}^{(k)})\sqrt{\frac{\log d^2}{|\mathcal{I}_{(i,j)}|^2}\sum_{k\in \mathcal{I}_{(i,j)}}\frac{1}{N_k}}\right).
\end{displaymath}
Since $\frac{1}{|\mathcal{I}_{(i,j)}|^2}\sum_{k\in \mathcal{I}_{(i,j)}}\frac{1}{N_k} \leq \frac{1}{|\mathcal{I}_{(i,j)}|N_{\min}}$, and by condition (a) of Assumption \ref{ass:aggregation}, $|\mathcal{I}_{(i,j)}|\geq 4K/7$, we write 
$|\mathcal{I}_{(i,j)}|=(1-e_{(i,j)})K$ for some $e_{(i,j)}\in(0,3/7]$, 
since $1-e_{(i,j)}\geq 4/7$ implies $e_{(i,j)}\leq 3/7$. Taking $e=\sup_{(i,j)\in[d]\times[dp]}e_{(i,j)}\in(0,3/7]$, the first term in 
\eqref{e:individual_sum} is bounded above by
\begin{displaymath}
    \mathcal{O}_{\mathbb{P}}\left(\max_{k\in[K]}\kappa(\Sigma_{\epsilon}^{(k)})
    \sqrt{\frac{\log d^2}{(1-e)KN_{\min}}}\right),
\end{displaymath}
where $e\in(0,3/7]$, and the bound on the first term in \eqref{e:individual_sum} dominates the upper bound of $|(\tilde{\alpha}_{i}^{(0)})_j - (\alpha_{i}^{(0)})_j|$.

For $k\in[K]$, by using the triangle inequality,
\begin{align}
    \|\tilde{\alpha}^{(k)} - \alpha^{(k)}\|_{\infty}
    &\leq \|\tilde{\alpha}^{(0)} - \alpha^{(0)}\|_{\infty} 
    + \|\tilde{\beta}^{(k)} - \beta^{(k)}\|_{\infty}  \nonumber\\
    &= \mathcal{O}_{\mathbb{P}}\Bigg(
    \kappa(\Sigma_{\epsilon}^{(k)})\sqrt{\frac{\log d^2}{N_k}} 
    + \max_{k\in[K]}\kappa(\Sigma_{\epsilon}^{(k)})\sqrt{\frac{\log d^2}{|\mathcal{I}_{(i,j)}|^2}\sum_{k\in \mathcal{I}_{(i,j)}}\frac{1}{N_k}} \nonumber\\
    &+ \kappa^2(\Sigma_{\epsilon}^{(k)})\kappa^4(\Phi^{(k)})|\!|\!|\Phi^{(k)}|\!|\!|^2\max\{s_{\Theta}^{(k)},s_{\Phi}^{(k)}\}\frac{\log d^2}{N_k} \nonumber\\
    & + \sum_{k=1}^K w_{i,j}^{(k)} \left(\kappa^2(\Sigma_{\epsilon}^{(k)})\kappa^4(\Phi^{(k)})|\!|\!|\Phi^{(k)}|\!|\!|^2\max\{s_{\Theta}^{(k)},s_{\Phi}^{(k)}\}\frac{\log d^2}{N_{k}}\right)\Bigg). \label{e:unique_bdd}
\end{align}
Note that the last two terms in \eqref{e:unique_bdd} vanish faster than the first two terms as $d,N_k$ grow. Hence, 
\begin{displaymath}
    \|\tilde{\alpha}^{(k)} - \alpha^{(k)}\|_{\infty} 
    = \mathcal{O}_{\mathbb{P}}\left(\kappa(\Sigma_{\epsilon}^{(k)})\sqrt{\frac{\log d^2}{N_k}}\right), \quad k\in[K].
\end{displaymath}
\end{proof}

\begin{lemma}\label{lem:sigma}
For each $k$, $\hat{\sigma}_{k,i}^2 \stackrel{p}{\to}\sigma_{k,i}^2$, $i\in[d]$.    
\end{lemma}

\begin{proof}
Note that for fixed $i$,
\begin{align}
    \hat{\sigma}_{k,i}^2 
    &= \frac{1}{N_k}\sum_{t=1}^{N_k}\left(\mathcal{Y}_{i,t}^{(k)} - \mathcal{X}_{t:}^{(k)}\hat{\beta}_{i}^{(k)}\right)^2 = \frac{1}{N_k}\|\mathcal{Y}_i^{(k)} - \mathcal{X}^{(k)}\hat{\beta}_i^{(k)}\|_2^2 \nonumber\\
    &= \frac{1}{N_k}\|\mathcal{X}^{(k)}\beta_i^{(k)} - \mathcal{X}^{(k)}\hat{\beta}_i^{(k)} + E_{i}^{(k)}\|_2^2 \nonumber\\
    &= \frac{1}{N_k}\| \mathcal{X}^{(k)}(\hat{\beta}_i^{(k)} - \beta_i^{(k)})\|_2^2 - \frac{2}{N_k}\langle \mathcal{X}^{(k)}(\hat{\beta}_i^{(k)} - \beta_i^{(k)}),E_{i}^{(k)} \rangle + \frac{1}{N}_k\|E_{i}^{(k)}\|_2^2. \label{e:sigma_est}
\end{align}
Note that by the law of large numbers, the last term in \eqref{e:sigma_est} converges to $\sigma_{k,i}^2$. For the first term in \eqref{e:sigma_est}, by using Proposition 3.3 in \cite{basu2015regularized_supp}, it is bounded above by
\begin{displaymath}
    \mathcal{O}_{\mathbb{P}}\left(\frac{s_{\Phi}^{(k)}(\lambda^{(k)}(\mathcal{M}(f_X^{(k)})+\mathcal{M}(f_{\epsilon}^{(k)})))^2}{m(f_X^{(k)})}\right) = \mathcal{O}_{\mathbb{P}}\left(\kappa^2(\Sigma_{\epsilon}^{(k)})\kappa^4(\Phi^{(k)})|\!|\!|\Phi^{(k)}|\!|\!|^2\frac{s_{\Phi}^{(k)}\log d}{N_k}\right) = \mathcal{O}_{\mathbb{P}}\left(\frac{1}{\sqrt{N_k}}\right).
\end{displaymath}
Note that the second term in \eqref{e:sigma_est} is, by using \eqref{e:deviation_bound} and \eqref{e:bdd_sparse} in Lemma \ref{lem:beta_hat_bdd}, is bounded above by 
\begin{align*}
    2\left\|\frac{1}{N_k}\mathcal{X}^{(k)'}E_{i}^{(k)}\right\|_{\infty}\|\hat{\beta}_i^{(k)} - \beta_{i}^{(k)}\|_{1} 
    &=
    \mathcal{O}_{\mathbb{P}}\left(\kappa^2(\Sigma_{\epsilon}^{(k)})\kappa^4(\Phi^{(k)})|\Phi^{(k)}||^2\frac{s_{\Phi}^{(k)}\log d}{N_k}\right)  \\
    &= \mathcal{O}_{\mathbb{P}}\left(\frac{1}{\sqrt{N_k}}\right).
\end{align*}
Hence, we have the desired result.    
\end{proof}

\begin{lemma}\label{lem:scaled_var}
$\hat{\sigma}_{i,k}\sqrt{\hat{\Omega}_{j,j}^{(k)}} \stackrel{p}{\to} \sigma_{i,k}\sqrt{\Theta_{j,j}^{(k)}}$ for all $i,j$, and $k$.    
\end{lemma}

\begin{proof}
Define $\hat{\Omega}^{(k)}= \hat{\Theta}^{(k)}\hat{\Sigma}^{(k)}\hat{\Theta}^{(k)'}$. Note that 
\begin{equation}
    \|\hat{\Omega}^{(k)} - \Theta^{(k)}\|_{\infty} = \|(\hat{\Theta}^{(k)}\hat{\Sigma}^{(k)}-I)\hat{\Theta}^{(k)'}\|_{\infty} + \|\hat{\Theta}^{(k)} - \Theta^{(k)}\|_{\infty}. \label{e:theta_est}
\end{equation}
The first term in \eqref{e:theta_est} is, by using \eqref{e:bias_covariance}, \eqref{e:max_tau_hat}, and $\|\gamma_j^{(k)}\|_1 \leq \mathcal{O}(\sqrt{(s_{\Theta}^{(k)})_j})$, bounded above by 
\begin{align*}
    \max_{j\in[d]} |e_j - \hat{\Theta}_{j:}^{(k)}\Sigma^{(k)}|\|\hat{\Theta}_{j:}^{(k)}\|_1 
    &\leq \max_{j\in[d]} |e_j - \hat{\Theta}_{j:}^{(k)}\Sigma^{(k)}|\|\hat{\gamma}_{j}^{(k)}/(\hat{\tau}_{j}^{(k)})^2\|_1\\
    &= \mathcal{O}_{\mathbb{P}}\left(\kappa(\Sigma_{\epsilon}^{(k)})\kappa^2(\Phi^{(k)})\sqrt{\frac{\log d}{N_k}}\right)\mathcal{O}(\max_{j\in[d]}\sqrt{(s_{\Theta}^{(k)})_j})\mathcal{O}_{\mathbb{P}}(1/\sigma_{k,\min}^2) \\
    &= \mathcal{O}_{\mathbb{P}}\left( \frac{\kappa(\Sigma_{\epsilon}^{(k)})\kappa^2(\Phi^{(k)})}{\sigma_{k,\min}^2}\sqrt{\frac{s_{\Theta}^{(k)}\log d}{N_k}}\right).
\end{align*}
The second term in \eqref{e:theta_est} is bounded above by
\begin{equation}\label{e:theta_difference_ell2}
    \max_{j\in[d]}\|\hat{\Theta}_{j:}^{(k)} - \Theta_{j:}^{(k)}\|_2 
    \leq \max_{j\in[d]}\left\{\|\hat{\gamma}_j^{(k)} - \gamma_j^{(k)}\|_2/(\hat{\tau}_j^{(k)})^2 + \|\gamma_j^{(k)}\|_2( 1/(\hat{\tau}_j^{(k)})^2 - 1/(\tau_j^{(k)})^2) \right\}.
\end{equation}
Fix $j$. Consider the second term in \eqref{e:theta_difference_ell2}. By using the result in the second term of \eqref{e:theta_difference_ell1} with $\|\gamma_j^{(k)}\|_2 = \mathcal{O}(\sqrt{(s_{\Theta}^{(k)})_j})$, it is bounded above by $o(1)$ for all $j\in[d]$. By using \eqref{e:max_tau_hat} and by Proposition 3.3 in \cite{basu2015regularized_supp}, the first term in \eqref{e:theta_difference_ell2} is bounded above by
\begin{displaymath}
    \mathcal{O}_{\mathbb{P}}\left(\frac{\sqrt{(s_{\Theta}^{(k)})_j}\lambda^{(k)}(\mathcal{M}(f_X^{(k)})+\mathcal{M}(f_{\epsilon}^{(k)}))}{m(f_X^{(k)})}\right)\mathcal{O}_{\mathbb{P}}(1/\sigma_{k,\min}^2).
\end{displaymath}
Combine those, we have
\begin{displaymath}
    \max_{j\in[d]}\|\hat{\Theta}_{j:}^{(k)} - \Theta_{j:}^{(k)}\|_2 
    = \mathcal{O}_{\mathbb{P}}\left(\frac{\kappa(\Sigma_{\epsilon}^{(k)})\kappa^2(\Phi^{(k)})|\!|\!|\Phi^{(k)}|\!|\!|^2}{\sigma_{k,\min}^2}\sqrt{\frac{s_{\Theta}^{(k)}\log d}{N_k}}\right).
\end{displaymath}
Hence, $\|\hat{\Omega}^{(k)} - \Theta^{(k)}\|_{\infty} = o_{\mathbb{P}}(1)$. By combining with Lemma \ref{lem:sigma}, it completes the proof.    
\end{proof}

\begin{lemma}\label{lem:martingale_clt}
For fixed $i,j$ and $k$, 
\begin{displaymath}
    \sum_{t=1}^{N_k}\frac{(\Theta_{j:}^{(k)'}X_{t-1}^{(k)})\varepsilon_{i,t}^{(k)}}{\sigma_{k,i}\sqrt{\Omega_{j,j}^{(k)}}} 
    \stackrel{d}{\to} \mathcal{N}(0,1).
\end{displaymath}

\end{lemma}

\begin{proof}
We omit $k$ for notational convenience. Define $Z_{t} =(\Theta_{j:}^{'}X_{t-1})\varepsilon_{i,t}$. Then for $S_n=\sum_{t=1}^NZ_t$, we want to show
\begin{equation}\label{e:clt}
    S_n:=\sum_{t=1}^NZ_{t} \stackrel{d}{\to} \mathcal{N}(0,\sigma_{j}^2\Omega_{j,j}).
\end{equation}
Denote $\mathcal{F}_N$ by a $\sigma$-field generated by $\{X_t,t\leq N\}$. Then
\begin{displaymath}
    \mathbb{E}[S_N-S_{N-1}|\mathcal{F}_N]=\mathbb{E}[Z_N|\mathcal{F}_N]=(\Theta_{j:}^{'}X_{N-1})\mathbb{E}[\varepsilon_{i,N}|\mathcal{F}_N]=0,
\end{displaymath}
so $\{Z_t\}$ is a martingale with respect to $\{F_N\}$. Claim that for
\begin{displaymath}
    V_N^2 = \sum_{t=1}^N \mathbb{E}[Z_t^2|\mathcal{F}_{N-1}] = \sum_{i=1}^N \sigma_{i}^2(\Theta_{j:}^{'}X_{t})^2 = N\sigma_i^2 \Theta_{j:}^{'}\hat{\Sigma}\Theta_{j:}
\end{displaymath}
and 
\begin{displaymath}
    v_n^2 = \mathbb{E}V_N^2 = N\sigma_i^2\mathbb{E}(\Theta_{j:}^{'}X_{t})^2 = N\sigma_i^2 \Theta_{j:}^{'}\Sigma\Theta_{j:},
\end{displaymath}
the following holds:
\begin{align}
    V_N^2/v_N^2 
    &\stackrel{d}{\to} 1, \label{e:first_condition} \\
    \frac{1}{v_N^2}\sum_{t=1}^N \mathbb{E}[Z_t^21_{\{|Z_t| > \epsilon v_N\}}] 
    & \to 0 \quad \textrm{for all }\epsilon>0. \label{e:second_condition}
\end{align}
Then, \eqref{e:clt} holds by the martingale central limit theorem \cite[e.g., Chapter 1.6 in][]{hall2014martingale_supp}. It remains to show that the claims hold. By Lemma \ref{lem:scaled_var}, the first claim \eqref{e:first_condition} is satisfied. Note that since $\{X_{t}\}$ and $\{\varepsilon_{i,t}\}$ are sub-Gaussian and $\Theta_{j,j}$ is bounded, 
\begin{displaymath}
    \|\Theta_{j:}'X_{t-1}\|_{\psi_2}\leq C_1, \quad \|\varepsilon_{i,t}\|_{\psi_2} \leq C_2\sigma_i,
\end{displaymath}
which implies that
\begin{displaymath}
    \|Z_{t}\|_{\psi_1}\leq C_1C_2\sigma_i,
\end{displaymath}
where $\|\cdot\|_{\psi_2}$, $\|\cdot\|_{\psi_1}$ are the sub-Gaussian and sub-exponential norms, respectively  \cite[e.g., Chapter 2 in][]{vershynin2018high_supp}. This gives $\mathbb{E}|Z_t|^{2+\epsilon_1}\leq C_3<\infty$ for $\epsilon_1>0$ \cite[e.g., Proposition 2.7.1 in][]{vershynin2018high_supp}. By H\"older's and Markov's inequalities,
\begin{align*}
    Z_t^21_{\{|Z_t| > \epsilon v_N\}} 
    &\leq (\mathbb{E}|Z_t|^{2+\epsilon_1})^{2/(2+\epsilon_1)}(\mathbb{P}(|Z_{t}|>\epsilon v_N))^{\epsilon_1/(2+\epsilon_1)} \\
    &\leq (\mathbb{E}|Z_t|^{2+\epsilon_1})^{2/(2+\epsilon_1)}\frac{\mathbb{E}|Z_t|^{2+\epsilon_1}}{(\epsilon v_N)^{2+\epsilon_1}} \\
    &\leq C_{3}^{2/(2+\epsilon_1)}\left(\frac{C_3}{\epsilon^{2+\epsilon_1} v_N^{2+\epsilon_1}}\right)^{\epsilon_1/(2+\epsilon_1)} \\
    &\leq \frac{C_4}{v_N^{\epsilon_1}}
\end{align*}
for some constant $C_4>0$. Since $v_N=\sqrt{N}\sigma_i\sqrt{\Theta_{j:}'\Sigma\Theta_{j:}}=\mathcal{O}(\sqrt{N})$, we have
\begin{displaymath}
    \frac{1}{v_N^2}\sum_{t=1}^N \mathbb{E}[Z_t^21_{\{|Z_t| > \epsilon v_N\}}] 
    \leq \frac{N}{v_N^2} \frac{C_4}{v_N^{\epsilon_1}} 
    = \frac{NC_4}{v_N^{2+\epsilon_1}} 
    = \mathcal{O}(N^{-\epsilon_1/2}) 
    \to 0
\end{displaymath}
as $N\to\infty$ since $\epsilon_1>0$. Hence, \eqref{e:second_condition} holds.
\end{proof}

\begin{lemma}\label{lem:clt_beta_pair}
For fixed $k,\ell\in[K]$, $k\neq\ell$, consider that we test the following hypotheses:
\begin{displaymath}
    H_{0,(i,j)}^{(k,\ell)}: (\beta_{i}^{(k)})_j = (\beta_{i}^{(\ell)})_j  \quad vs \quad 
    H_{1,(i,j)}^{(k,\ell)}: (\beta_{i}^{(k)})_j \neq (\beta_{i}^{(\ell)})_j.
\end{displaymath}
Then, we have
\begin{equation}\label{e:clt_beta_pair}
    Z_{(i,j)}^{(k,\ell)}((\beta_{i}^{(k)})_j-(\beta_{i}^{(\ell)})_j) 
    = \frac{((\tilde{\beta}_{i}^{(k)})_j-(\tilde{\beta}_{i}^{(\ell)})_j) - ((\beta_{i}^{(k)})_j-(\beta_{i}^{(\ell)})_j)}{\sqrt{\frac{1}{N_{k}}\hat{V}_{i,j}^{(k)} + \frac{1}{N_{\ell}}\hat{V}_{i,j}^{(\ell)}}}
    \stackrel{d}{\to} \mathcal{N}(0,1).
\end{equation}
Under $H_{0,(i,j)}^{(k,\ell)}$, we reject the null hypothesis if $|Z_{(i,j)}^{(k,\ell)}(0)|\geq z_{1-\alpha/2}$ for a given $\alpha$. 

\end{lemma}

\begin{proof}
From \eqref{e:clt_betak}, we have for each subject separately,
\begin{displaymath}
    \sqrt{N_k}((\tilde{\beta}_{i}^{(k)})_j - (\beta_{i}^{(k)})_j) \stackrel{d}{\to} \mathcal{N}(0,V_{i,j}^{(k)}), \quad
    \sqrt{N_{\ell}}((\tilde{\beta}_{i}^{(\ell)})_j - (\beta_{i}^{(\ell)})_j) \stackrel{d}{\to} \mathcal{N}(0,V_{i,j}^{(\ell)}).
\end{displaymath}
Since the two estimators are independent, the joint convergence also holds. By the continuous mapping theorem applied to the linear map $(a,b)\to a/\sqrt{N_k}-b/\sqrt{N_{\ell}}$, 
\begin{displaymath}
    ((\tilde{\beta}_{i}^{(k)})_j - (\beta_{i}^{(k)})_j) - ((\tilde{\beta}_{i}^{(\ell)})_j - (\beta_{i}^{(\ell)})_j) \stackrel{d}{\to}\mathcal{N}\left(0,\frac{V_{i,j}^{(k)}}{N_k} + \frac{V_{i,j}^{(\ell)}}{N_{\ell}}\right).
\end{displaymath}
From Lemma \ref{lem:scaled_var}, by using the continuous mapping theorem,
\begin{displaymath}
    \frac{\hat{V}_{i,j}^{(k)}}{N_k} + \frac{\hat{V}_{i,j}^{(\ell)}}{N_{\ell}}
    \stackrel{p}{\to} \frac{V_{i,j}^{(k)}}{N_k} + \frac{V_{i,j}^{(\ell)}}{N_{\ell}}.
\end{displaymath}
Hence, by Slutsky's theorem, we have the result as \eqref{e:clt_beta_pair}.
\end{proof}

\section{Additional Simulation Results}\label{se:additional}

In this section, we describe the performance metrics and 
present additional simulation results.

\subsection{Performance Metrics for Simulation Study}\label{sse:additional_setups}

First, we describe the performance metrics used to evaluate the estimation results. For the true $(\alpha_i^{(0)})_j,\{\alpha_i^{(k)})_j\}_{k\in[K]}$ and their estimators $(\hat{\alpha}_i^{(0)})_j,\{\hat{\alpha}_i^{(k)})_j\}_{k\in[K]}$, we compute
\begin{equation*}
    \begin{array}{lll}
    & \textrm{RMSE}(\alpha^{(0)}) = \frac{\|\hat{\alpha}^{(0)} - \alpha^{(0)}\|_2}{\|\alpha^{(0)}\|_2}, 
    & \textrm{RMSE}(\alpha^{(K)}) = \frac{1}{K}\sum_{k=1}^K\frac{\|\hat{\alpha}^{(k)} - \alpha^{(k)}\|_2}{\|\alpha^{(k)}\|_2}, \\
    
    & \textrm{Sens}(\alpha^{(0)}) = \frac{\sum_{i,j}1_{\{(\hat{\alpha}_i^{(0)})_j \neq 0 \ \& \ (\alpha_i^{(0)})_j \neq 0\}}}{\sum_{i,j}1_{\{(\alpha_i^{(0)})_j \neq 0\}}},
    & \textrm{Sens}(\alpha^{(K)}) = \frac{1}{K}\sum_{k=1}^K\frac{\sum_{i,j}1_{\{(\hat{\alpha}_i^{(k)})_j \neq 0 \ \& \ (\alpha_i^{(k)})_j \neq 0\}}}{\sum_{i,j}1_{\{(\alpha_i^{(k)})_j \neq 0\}}} \\
    & \textrm{Spec}(\alpha^{(0)}) = \frac{\sum_{i,j}1_{\{(\hat{\alpha}_i^{(0)})_j = 0 \ \& \ (\alpha_i^{(0)})_j = 0\}}}{\sum_{i,j}1_{\{(\alpha_i^{(0)})_j = 0\}}},
    & \textrm{Spec}(\alpha^{(K)}) = \frac{1}{K}\sum_{k=1}^K\frac{\sum_{i,j}1_{\{(\hat{\alpha}_i^{(k)})_j = 0 \ \& \ (\alpha_i^{(k)})_j = 0\}}}{\sum_{i,j}1_{\{(\alpha_i^{(k)})_j = 0\}}}.
    \end{array}
\end{equation*}

Next, we describe how to compute the empirical FDR and power. To define the true null and alternative sets $\mathcal{H}_0,\mathcal{H}_1$, we define four mutually disjoint index sets, $\mathcal{T}_{0K}$, $\mathcal{T}_{0^cK}$, $\mathcal{T}_{0K^c}$, and $\mathcal{T}_{0^cK^c}$, over the set of interests $\mathcal{S}=\{(i,j)\in[d] \times [dp]\}$. These sets correspond to the following distinct profiles, respectively: (i) both common and unique paths are nonzero, (ii) common paths are zero but unique paths are nonzero, (iii) common paths are nonzero but unique paths are zero, and (iv) both common and unique paths are zero. By construction, these sets satisfy $\mathcal{S}=\mathcal{T}_{0K} \cup \mathcal{T}_{0^cK} \cup \mathcal{T}_{0K^c} \cup \mathcal{T}_{0^cK^c}$. The true alternative ($\mathcal{H}_1$) and true null ($\mathcal{H}_0$) configurations for each framework are then specified as follows:
\begin{itemize}
    \item[\textbf{(a)}] 
    Test of significance: 
    $\mathcal{H}_1=\mathcal{H}_1^{(0)}=\mathcal{T}_{0K}\cup\mathcal{T}_{0K^c}$ and $\mathcal{H}_0=\mathcal{H}_0^{(0)}=\mathcal{T}_{0^cK}\cup\mathcal{T}_{0^cK^c}$.
    \item[\textbf{(b)}] 
    Test of homogeneity: 
    $\mathcal{H}_1=\mathcal{H}_1^{[K]}=\mathcal{T}_{0K}\cup\mathcal{T}_{0^cK}$ and $\mathcal{H}_0=\mathcal{H}_0^{[K]}=\mathcal{T}_{0K^c}\cup\mathcal{T}_{0^cK^c}$.
    \item[\textbf{(c)}] 
    Test of nullity: $\mathcal{H}_1=\mathcal{H}_1^{[K]_{\circ}}=\mathcal{T}_{0K}\cup\mathcal{T}_{0^cK}\cup\mathcal{T}_{0K^c}$ and 
    $\mathcal{H}_0=\mathcal{H}_0^{[K]_{\circ}}=\mathcal{T}_{0^cK^c}$.
\end{itemize}
Finally, the computation of the FDR and power is completed by determining the threshold $\hat{\zeta}$ in \eqref{e:thres_sig} and constructing the corresponding sets of discoveries $\hat{\mathcal{S}}$ in \eqref{e:discovery_set} for each of the three hypothesis tests.

\subsection{Simulation Results for Higher-Dimensional Setups}\label{sse:additional_higher}

\begin{figure}[h]
 \centering
 \includegraphics[width=0.7\textwidth,height=0.4\textheight]
 {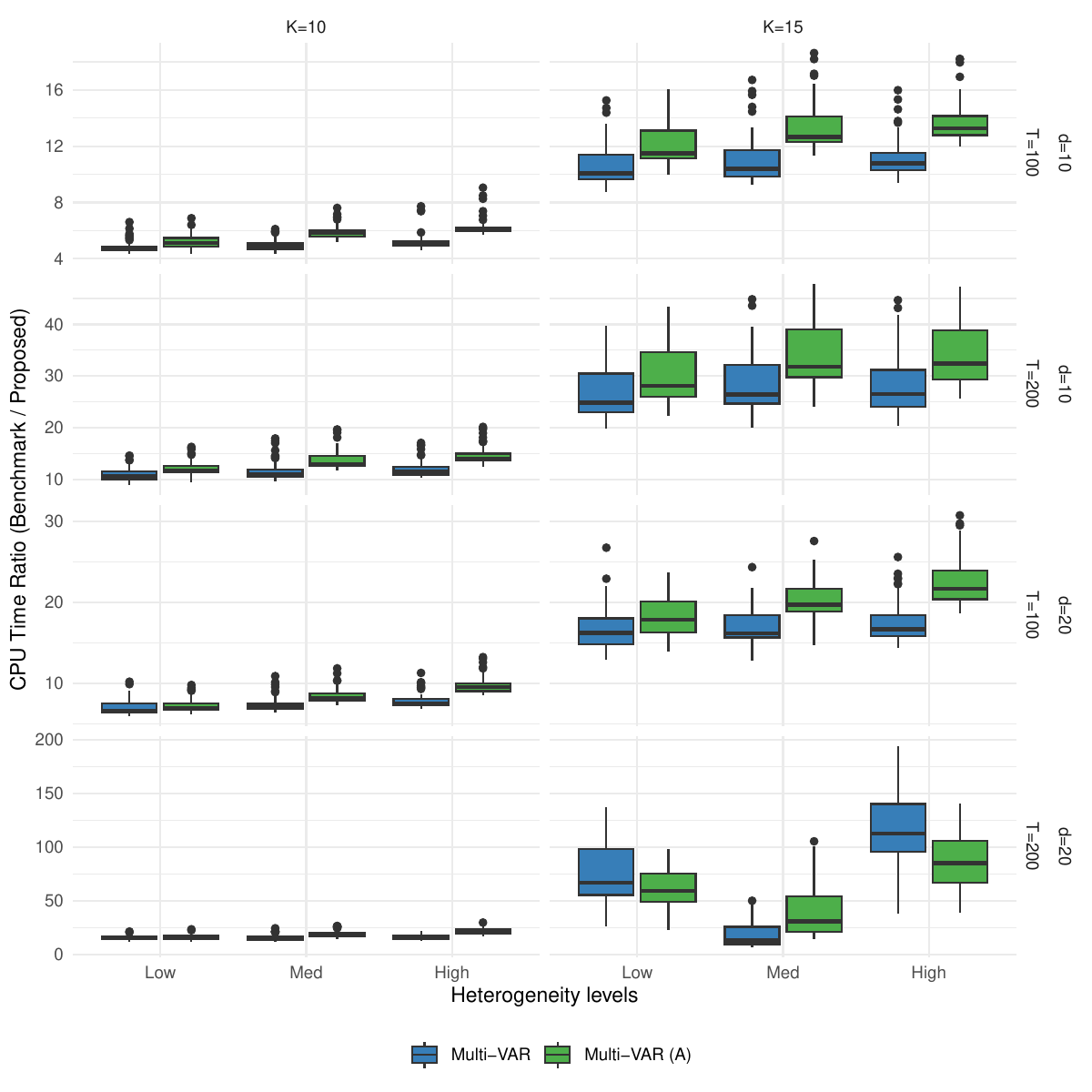}
 \caption{Boxplots of the CPU time ratio (Benchmark / Proposed) under different combinations of $d$ and average $T$ (combinations indicated on the right tabs), $K$ (each column), and heterogeneity levels (each axis).}
 \label{fig:sim_cputime}
\end{figure}

Figure \ref{fig:sim_cputime} displays the CPU time ratio ($\text{Benchmark} / \text{Proposed}$). Because this ratio tends to escalate dramatically with problem size, a comprehensive comparative simulation against the multi-VAR and multi-VAR (A) is computationally intractable in higher dimensions in Section \ref{sse:result_estimation}. Therefore, this section focuses exclusively on the simulation results for the proposed multi-VARSE framework under higher-dimensional setups.

We additionally conduct simulations with $d\in\{30,50\}$. As there were no substantial differences in performance metrics across varying numbers of subjects, we restrict our focus to the case where $K=10$. Furthermore, we modify the simulation setup by dropping the fixed sparsity levels for the common and unique paths. As the dimension $d$ grows, it becomes intractable to randomly assign non-zero entries while ensuring that every subject's VAR model satisfies the stationarity condition. We adjust the density parameters to $(s_0,s_k) \in \{(0.04,0.01),(0.03,0.02),(0.01,0.04)\}$, which correspond to low, medium, and high levels of heterogeneity, respectively. All other experimental settings remain identical to the configurations described in Section \ref{sse:simulation_setup}.

\begin{figure}[th!]
 \centering
 \includegraphics[width=1\textwidth,height=0.5\textheight]
 {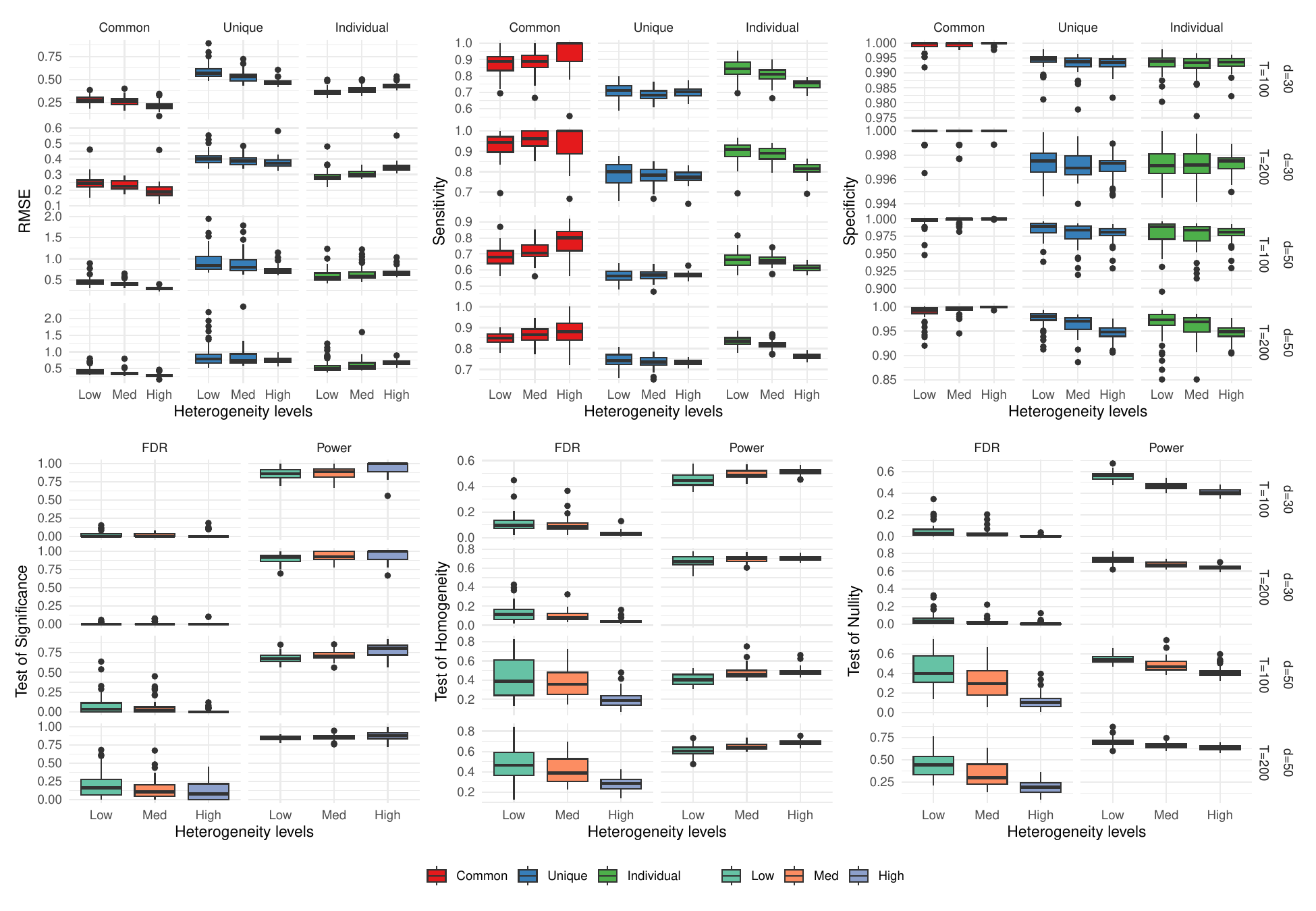}
 \caption{Boxplots of the root mean square error (RMSE), sensitivity, and specificity for $\alpha^{(0)}$ and the average of $\alpha^{(k)}$ (top row), alongside boxplots of the FDR and power for the tests of significance, homogeneity, and nullity (bottom row), from the proposed framework. The panels are arranged vertically by dimension $d$ and sample size $T$ (combinations indicated on the right tabs), and horizontally by heterogeneity levels (on the $x$-axes), combined either with components (top) or type of measures (bottom).}
 \label{fig:sim_higher}
\end{figure}

Figure \ref{fig:sim_overlap} presents the simulation results for the higher-dimensional settings. For estimation performance, the higher dimensionality is expected to degrade the overall quality of the estimators compared to the lower-dimensional settings, which is reflected in generally elevated RMSE values. Nevertheless, this RMSE for common, unique, and individual paths steadily decreases as the sample size increases. Furthermore, sensitivity for common paths shows a gradual increase from low to high heterogeneity levels, while specificity remains preserved decent across all setups.  Regarding inference, while the test of significance maintains an FDR near zero and high power across all dimensions, the tests of homogeneity and nullity perform worse than in lower dimensions. This deterioration stems from poorer estimation quality in higher dimensions; since the underlying model is highly sparse, the number of true positive entries is inherently small, making the FDR prone to inflation. This implies that, given fewer true positives, the detrimental effect of degraded estimation in high dimensions outweighs any statistical benefits gained from a larger number of variables. Consequently, as the dimension grows, the FDR for both tests increases and power declines, though this inflation is mitigated when the sample size increases.

\subsection{Additional Results in Different Simulation Setup}\label{sse:additional_different}

\begin{figure}[th!]
 \centering
 \includegraphics[width=1\textwidth,height=0.5\textheight]
 {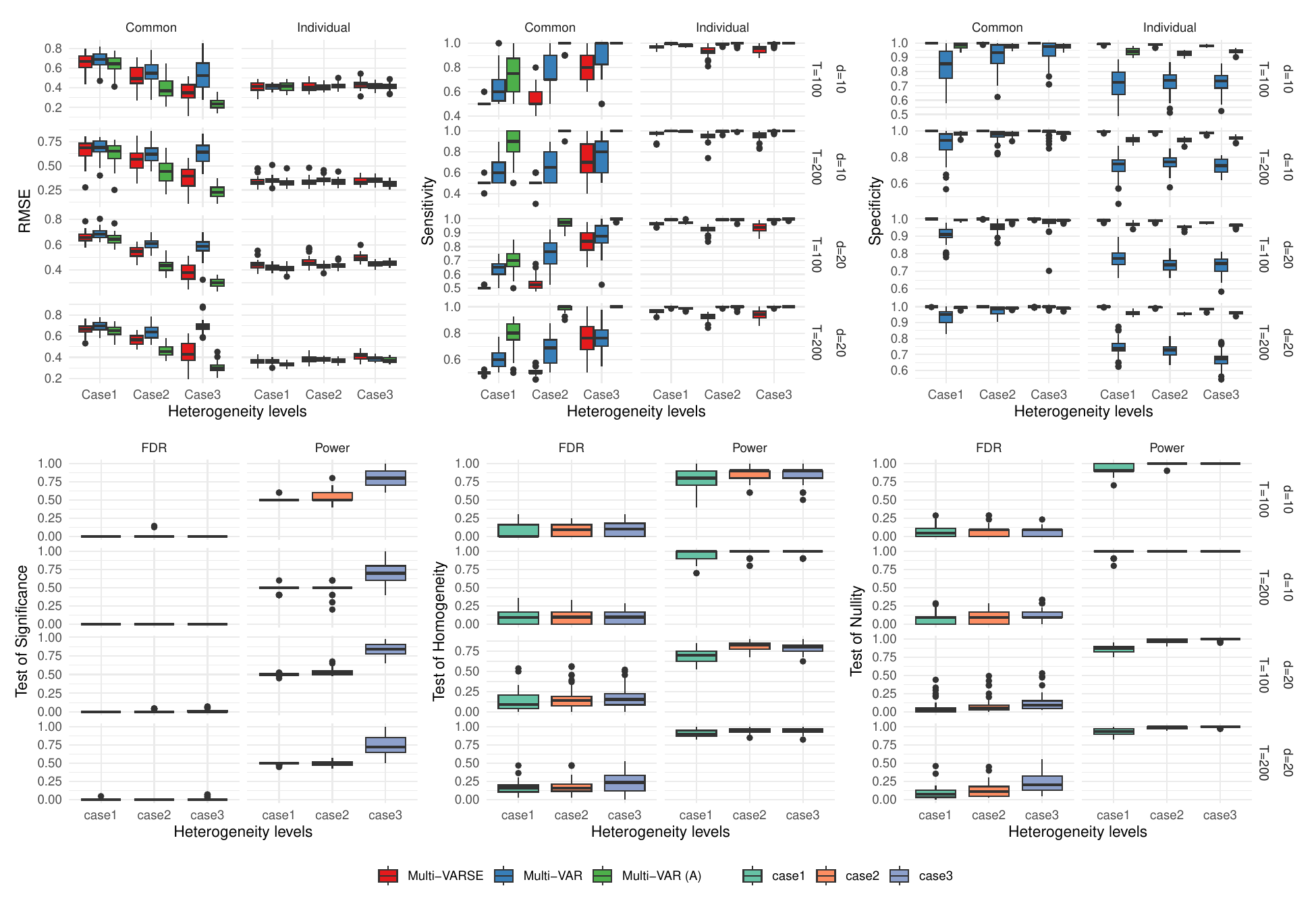}
 \caption{Boxplots of the root mean square error (RMSE), sensitivity (Sens), and specificity (Spec) of $\alpha^{(0)}$ and the average of $\beta^{(k)}$ (top), alongside boxplots of the FDR and power for the tests of significance, homogeneity, and nullity (bottom).}
 \label{fig:sim_overlap}
\end{figure}

In this experiment, we generate data following the process outlined in \cite{fisher2024structured_supp}. While the data-generation scheme in Section \ref{sse:simulation_setup} is suitable for multi-VAR modeling, it may be unrealistic for scenarios where certain dynamic paths are shared only among specific subsets of individuals. To mirror heterogeneity more realistically, we introduce the two-stage data-generating framework.  In Stage 1, we select all subjects and randomly sample $5\%$ of the individual paths, yielding $s_0=0.05$. In Stage 2, we select $20\%$ of subjects to receive an additional $5\%$ of subject-specific non-zero entries, so that $s_k = 0.05$, which are non-overlapping with the Stage 1 selection. We call it Case 1. In Case 2 and Case 3, we increase the ratio of subjects selected by $40\%$ and $60\%$, respectively. Throughout the experiment, we fix $K=10$, and vary the dimension $d\in\{10,20\}$ and the average sample length $T\in\{100,200\}$. The remaining configurations are identical to the setup in Section \ref{sse:simulation_setup}.

Figure \ref{fig:sim_overlap} presents the simulation results for both estimation and inference. Notably, all three methods struggle to accurately identify the common paths, whereas their individual estimators exhibit the desired performance. Regarding all three performance metrics, the multi-VAR method with the adaptive scheme delivers the top performance, followed by the proposed method. Although the proposed framework exhibits lower specificity compared to the original multi-VAR approaches, its RMSE remains competitive and shows substantial improvement as the heterogeneity level increases. This highlights that Multi-VARSE tends to be more conservative when identifying nonzero common paths, as shown in its lower sensitivity and higher specificity.

Regarding hypothesis testing, the test of significance achieves near-perfect control, preserving low FDR and maintaining high power across all cases. Conversely, the tests of homogeneity and nullity are relatively vulnerable to the level of heterogeneity. Nonetheless, enlarging the sample sizes counters this degradation, reducing the FDR and stabilizing power across all setups. Overall, while applying the proposed method to accommodate this alternative definition of common paths presents distinct challenges, these results are promising, as the performance of the proposed framework exhibits patterns similar to those shown in the standard setup of Section \ref{sse:simulation_setup}.

\end{document}